\documentclass[11pt]{article}
\usepackage{amssymb,amsmath,amsthm,mathtools}
\usepackage{array,geometry}  
\usepackage{setspace}
\usepackage{bm,paralist,dsfont,nicefrac}
\usepackage{etoolbox}
\usepackage[normalem]{ulem}
\usepackage{graphicx}
\usepackage{tabularx}
\usepackage{multicol}

\usepackage{tikz,float,booktabs,pifont,multirow,subcaption}
\usepackage{pgfplots}
\pgfplotsset{compat=1.18}
\usepackage{makecell}
\usetikzlibrary{fit,calc,math,shapes,shapes.multipart,decorations.text,arrows,decorations.markings,decorations.pathmorphing,shapes.geometric,positioning,decorations.pathreplacing, patterns}

\usepackage{natbib}
\defcitealias{DMM+21tight}{Dudycz, Manurangsi, Marcinkowski, and Sornat (2020)}

\usepackage[pagebackref,colorlinks]{hyperref}
\hypersetup{linkcolor=[rgb]{.7,0,0}}
\hypersetup{citecolor=[rgb]{.2,.7,.2}}
\hypersetup{urlcolor=[rgb]{.7,0,.7}}
\renewcommand*{\backref}[1]{}
\renewcommand*{\backrefalt}[4]{%
    \ifcase #1 (Not cited.)%
    \or        (Cited on page~#2)%
    \else      (Cited on pages~#2)%
    \fi}

\usepackage[capitalize, nameinlink, noabbrev]{cleveref}
\AddToHook{cmd/appendix/before}{%
	\crefalias{section}{appendix}%
	\crefalias{subsection}{appendix}
}

\usepackage{xcolor}

\usepackage{thmtools,thm-restate}
\usepackage{enumitem}
\usepackage{nicematrix}

\allowdisplaybreaks

\newcommand{\A}{\mathcal{A}}
\newcommand{\APXH}{\textrm{\textup{APX-hard}}}

\newcommand{\coNPC}{\textrm{\textup{coNP-complete}}}
\newcommand{\E}{\mathbb{E}}
\newcommand{\EJR}{\textrm{\textup{EJR}}}
\newcommand{\EJRplus}{\textrm{\textup{EJR$+$}}}
\newcommand{\eps}{\varepsilon}
\renewcommand{\epsilon}{\varepsilon}
\newcommand{\fPO}{\textrm{\textup{fPO}}}
\newcommand{\I}{\mathcal{I}}
\newcommand{\JR}{\textrm{\textup{JR}}}
\newcommand{\localPAV}[1]
{#1-\textup{local-PAV}}
\newcommand{\N}{\mathbb{N}}
\newcommand{\NP}{\textrm{\textup{NP}}}
\newcommand{\NPH}{\textrm{\textup{NP-hard}}}

\DeclareMathOperator{\OPT}{\textsf{\textup{OPT}}}
\DeclareMathOperator{\PAV}{\textsf{\textup{PAV}}}
\newcommand{\PO}{\textup{PO}}
\newcommand{\Poi}{\textsf{\textup{Poi}}}
\newcommand{\poly}{\textrm{\textup{Poly-time}}}

\DeclareMathOperator{\supp}{supp}

\newcommand{\w}{\mathbf{w}}

\Crefname{algocf}{Algorithm}{Algorithms}

\usepackage[linesnumbered,ruled,vlined]{algorithm2e}

\newtheorem{theorem}{Theorem}[section]
\newtheorem*{theorem*}{Theorem}

\newtheorem{definition}[theorem]{Definition}
\newtheorem*{definition*}{Definition}

\newtheorem{lemma}[theorem]{Lemma}
\newtheorem*{lemma*}{Lemma}

\newtheorem*{claim*}{Claim}

\newtheorem*{fact*}{Fact}

\newtheorem*{observation*}{Observation}

\newtheorem*{conjecture*}{Conjecture}

\newtheorem*{corollary*}{Corollary}

\newtheorem{remark}[theorem]{Remark}
\newtheorem*{remark*}{Remark}

\newtheorem{proposition}[theorem]{Proposition}
\newtheorem*{proposition*}{Proposition}

\newtheorem{example}[theorem]{Example}
\newtheorem*{example*}{Example}

\title{Finding Representative and Approximately Efficient Committees}
\usepackage{authblk}

\author[1]{Dominik Peters}
\author[2]{Rohit Vaish}
\author[2]{Jatin Yadav}
\affil[1]{CNRS, LAMSADE, Universit\'e Paris Dauphine - PSL; \texttt{mail@dominik-peters.de}}
\affil[2]{Indian Institute of Technology Delhi; \texttt{\{rvaish,csz237549\}@iitd.ac.in}} 
\date{\vspace{-1cm}}

\usepackage{colortbl}
\colorlet{myred}{red!25}
\colorlet{myblue}{blue!25}
\colorlet{mygreen}{green!25}

\usepackage{todonotes}

\begin{document}

\maketitle

\begin{abstract}
In approval-based committee voting, proportional approval voting (PAV) is a well-studied rule that combines proportional representation with Pareto efficiency. However, computing a PAV committee is \NPH{}, raising a natural question: Can the proportionality and efficiency properties of PAV be achieved via computationally efficient procedures?

We make two contributions toward answering this question. First, building on the known proportionality guarantees of the local-search-based variant of PAV (or local PAV), we systematically study its efficiency properties. We show that local PAV committees are weakly Pareto optimal, meaning that no other committee is strictly preferred by every voter. We also identify limitations: Local PAV guarantees only a $2$-approximation to fractional Pareto optimality ($2$-\fPO{}) and a $2/3$-approximation to the optimal PAV score, and both bounds are tight. In contrast, global PAV is Pareto optimal and satisfies the stronger $\alpha^\star$-\fPO{} guarantee, where $\alpha^\star \approx 1.346$ is the unique solution of $\int_0^{\alpha^\star} \frac{1-e^{-y}}{y} \, dy = 1$, and this approximation is tight.

Second, we design a polynomial-time algorithm that combines the best of these guarantees. The committee returned by our algorithm satisfies \EJRplus{} (a proportionality guarantee), $\alpha^\star$-\fPO{}, and weak Pareto optimality. It also achieves a $0.79$-approximation to the optimal PAV score, matching the best possible polynomial-time approximation assuming $\textrm{\textup{P}} \neq \NP{}$. Our algorithm works by pipage rounding a concave relaxation of the PAV objective and using that committee to initialize local PAV, thereby combining global approximation guarantees with local search stability.
\end{abstract}

\section{Introduction}
\label{sec:Introduction}

In approval-based committee elections \citep{LS23book}, our task is to select a committee $W \subseteq C$ of $k$ candidates from a set $C$ of candidates, based on approval preferences of a set $N$ of $n$ voters, where each voter $i \in N$ reports a subset $A_i \subseteq C$ of candidates that $i$ approves. One of the best-known voting rules in this setting was proposed by \citet{Thie1895}, and is now known as \emph{Proportional Approval Voting} (PAV). It is based on global optimization, and selects the committee $W \subseteq C$ with $|W| = k$ maximizing the objective function
\[
\sum_{i = 1}^n \left(1 + \frac12 + \frac13 + \dots + \frac{1}{|W \cap A_i|}\right).
\]
Thus, PAV optimizes the sum of the harmonic numbers induced by voter utilities $|W \cap A_i|$ corresponding to the number of winning candidates that they approve. This rule has been shown to satisfy formal properties that encode the idea of \emph{proportional representation}, where groups of voters with similar preferences should be represented in the outcome by a number of candidates proportional to the group's size. A classic formalization of this idea is \emph{Extended Justified Representation} (\EJR{}) \citep{ABC+17justified} and its recent strengthening \EJRplus{} \citep{BP23robust}.

The PAV rule also satisfies the efficiency criterion called \emph{Pareto optimality} (PO), which requires that there is no alternative committee $W'$ that all voters weakly prefer over $W$ (i.e., $|W' \cap A_i| \ge |W \cap A_i|$ for all $i \in N$), with at least one voter strictly preferring $W'$ (i.e., $|W' \cap A_i| > |W \cap A_i|$ for some $i \in N$). While \EJRplus{} is satisfied by several other voting rules (such as the Method of Equal Shares; \citealp{PS20proportionality}), PAV is the only known natural rule that satisfies both \EJRplus{} and \PO{}. However, computing the optimal PAV committee is \NPH{}~\citep{AGG+15computationalaspects}. Indeed, there is currently no known polynomial-time algorithm for computing a committee that satisfies both \EJRplus{} and \PO{}. Therefore, we will investigate which proportionality and efficiency properties of PAV can be achieved through an efficient algorithm.

\subsection*{Measuring Efficiency}

To compare different algorithms, we need a way of measuring their efficiency. Pareto optimality is one option, but we will also consider others. In particular, Pareto optimality can be relaxed to \emph{weak} Pareto optimality, which only requires that there not be a committee $W'$ that \emph{all} voters strictly prefer to $W$. One can also relax Pareto optimality using a multiplicative approximation, so that $W$ is $\alpha$-\PO{} ($\alpha \ge 1$) if there is no committee $W'$ such that $|W' \cap A_i| \ge \alpha \cdot |W \cap A_i|$ for all voters $i \in N$ with at least one inequality strict.

Instead of relaxing Pareto optimality, we can also \emph{strengthen} it by considering the concept of \emph{fractional Pareto optimality} (\fPO{}), recently studied in the context of approval-based committee elections by \citet{BBF+26fractional}. A committee satisfies \fPO{} if it is not even dominated by a \emph{fractional committee} (in which candidates can be fractionally part of the committee). Although it may seem counterintuitive to consider a stronger property when \PO{} is already challenging to guarantee, the motivation to consider \fPO{} lies in its better algorithmic behavior. Specifically, we can verify in polynomial time whether a committee is \fPO{} by using a linear program, whereas checking \PO{} is known to be \coNPC{}~\citep{AM20computing}. 

Additionally, it is interesting to consider a multiplicative approximation, $\alpha$-\fPO{}, which forbids the existence of a fractional committee that improves each voter's utility by at least an $\alpha$ factor. Notably, the PAV rule does not satisfy \fPO{}~\citep{BBF+26fractional}, but as we will show, it satisfies $1.346$-\fPO{}---indicating that we can use the approximation factor of an algorithm as a metric for its efficiency. Beyond its algorithmic appeal, $\alpha$-\fPO{} may also be normatively desirable in certain contexts. \PO{} can be seen not only as an efficiency notion but also as a stability notion \citep[cf.][]{EFF20pricepareto}, in that if a committee $W$ is Pareto dominated by $W'$, then voters have a justified complaint against $W$. However, voters may also complain if they can identify a lottery over committees that makes each of them better off ex ante compared to the chosen committee $W$. Especially in cases where the ex-ante utility goes up by a large factor for every voter, voters may be willing to bear the risk inherent in a random choice.\footnote{Note that each fractional committee can be represented as a lottery over integral committees~\citep{ALM+19strategyproof}.}
Another motivation for approximate \fPO{} comes from the \emph{best-of-both-worlds} literature on lotteries over deterministic committees, which seeks to combine ex-ante and ex-post guarantees~\citep{aziz2023bestofbothworlds,SV24maximumflowfairnetwork,jin2026limitationsbobw}. In order to combine approximate ex-ante efficiency with ex-post representation, it is necessary to be able to compute committees that satisfy $\alpha$-\fPO{}, not just $\alpha$-\PO{}.

A final measure of efficiency that we will consider is how well an algorithm approximates the optimal PAV score. While it is not clear that a committee with a higher PAV score is \emph{always} better than one with a lower score (especially since PAV's attractive properties like \EJRplus{} can fail even for committees with high PAV scores), often this will be the case: a higher PAV score generally indicates higher voter utilities, especially among voters with low utility. For a similar reason, approximations of the Nash product objective have been studied in the theory of fair allocation~\citep[e.g.,][]{CG18nashsocialwelfare,BKV18ef1}.

\subsection*{Our Contributions}

\newcommand{\newresult}{\cellcolor{blue!15}}
\begin{table}[t]
	\centering
	\small
	\begin{tabular}{ccccccc}
		\toprule
		Voting Rule & Running time & Fairness & $\alpha$-\fPO{} & $\beta$-\PO{} & Weak \PO{} & $\gamma$-$\PAV$ \\
		\midrule
		PAV & \APXH{} & \EJRplus{} & \newresult $\alpha^\star \approx 1.346$ & $1$ & Yes & $1$\\
		& & & & & & \\
		Exact local PAV ($\tau=0$) & Superpolynomial$^\dagger$ & \EJRplus{} & \newresult $2$ & \newresult $1+\eps$ & \newresult Yes & \newresult $2/3$\\
		& & & & & & \\
		\localPAV{$\tau$}, $\tau = \frac{1}{2k^2}$ & \poly{} & \EJRplus{} & \newresult $2$ & \newresult $1+\eps$ & \newresult Yes & \newresult $ 2/3$\\
		& & & & & & \\
		DMMS algorithm & \poly{} & \newresult Fails \JR{} & \newresult $1$ & \newresult $1$ & \newresult Yes & $\gamma^\star \approx 0.79$\\
		\citep{DMM+21tight} & & & & & & \\
		& & & & & & \\
		
		Our algorithm (Thm.~\ref{thm:Main-Algo}) & \poly{} & \newresult \EJRplus{} & \newresult $\alpha^\star \approx 1.346$ & \newresult $1+\eps$ & \newresult Yes & \newresult $\gamma^\star \approx 0.79$\\
		
		\bottomrule
	\end{tabular}
	\caption{Summary of known and contributed results (the latter are highlighted in shaded boxes). Each row corresponds to a voting rule, while the columns display the best-known bounds on the worst-case running time, representation axiom, weak \PO{}, and multiplicative approximation to \fPO{}, \PO{}, and the optimal PAV score. Here, $\eps > 0$ is arbitrary. The $\dagger$ entry corresponds to the $\Omega(k^{\log k})$ lower bound of~\citet{KE24lower} for lexicographic better-response dynamics.}
	\label{tab:Results}
\end{table}

The paper makes three main contributions: First, we analyze the efficiency properties of the global PAV rule and its local search-based variant called \emph{local PAV}~(\Cref{subsec:Results_local_PAV,subsec:Results_global_PAV}). Second, we present a polynomial-time \emph{round-and-swap} algorithm based on pipage rounding that meets the representation and efficiency guarantees of the PAV rule~(\Cref{subsec:Results_Our_Algorithm}). Third, we outline several natural proof strategies for combining representation (\JR{} or \EJRplus{}) with exact \fPO{} and explain their limitations~(\Cref{appendix:Limitations}). \Cref{tab:Results} provides an overview of our positive results.

\paragraph{Efficiency Properties of Local PAV.} 
A \localPAV{$\tau$} committee is one for which no swap of a winning candidate and a losing candidate increases the PAV objective by at least $\tau$. When $\tau \le \frac{n}{k^2}$, the \localPAV{$\tau$} rule satisfies essentially the same proportionality properties as global PAV, including \EJRplus{}. Additionally, for $\tau \geq \frac{1}{\textup{poly}{(n,k)}}$, a \localPAV{$\tau$} committee can be found in polynomial time through local search.

While \localPAV{$\tau$} committees can fail to satisfy \PO{} for any $\tau \ge 0$, we establish the surprising result that whenever $\tau \le \frac{n}{k^2}$, \localPAV{$\tau$} committees satisfy weak \PO{}, and for $\tau \le \frac{1}{k^2}$, they satisfy $(1+\epsilon)$-\PO{} for every $\epsilon > 0$~(\Cref{tab:Local-PAV-Results}). Thus, despite being local optima, local PAV committees satisfy global efficiency properties. However, local PAV only provides a $2$-\fPO{} guarantee in the worst case, compared to the $1.346$-\fPO{} guarantee of global PAV~(\Cref{thm:global_PAV_properties}). Furthermore, there exist instances where the PAV score of a local PAV committee is at most $\frac{2}{3}$ of the optimal PAV score, which is below the $0.79$-approximation that is known to be achievable in polynomial time~\citep{DMM+21tight}.

\paragraph{Rounding-Based Algorithm.}

In \Cref{subsec:Results_Our_Algorithm}, we develop a polynomial-time algorithm whose output is a committee satisfying \EJRplus{}, weak \PO{}, $1.346$-\fPO{}, and a $0.79$-approximation to the PAV score. Thus, our algorithm meets the fairness and efficiency properties of PAV while circumventing the associated computational intractability.

Our algorithm (\Cref{alg:Pipage-and-swap}) involves three main steps: First, it computes a fractional committee $\widehat{x}$ that approximately maximizes a concave function (to within an inverse-polynomial additive error) of voters' utilities $\Psi(x) \coloneqq \sum_i h(u_i(x))$, where $h(z) \coloneqq \int_0^z \frac{1-e^{-y}}{y} \, dy$. The function $\Psi(\cdot)$ provides a smooth lower bound to the multilinear extension of the PAV score $F_{\PAV}$. Second, it applies pipage rounding to $\widehat{x}$ (with respect to $F_{\PAV}$) to compute an integral committee $W_0$; doing so ensures that the PAV score of $W_0$ is at least $\Psi(\widehat{x})$. However, the committee $W_0$ may not satisfy \EJRplus{}. To fix this, the third part of the algorithm runs a local search phase, wherein, starting from the committee $W_0$, it swaps specific elected and unelected candidates to improve the PAV score. (Specifically, starting from $W_0$, our algorithm performs PAV improving swaps until it reaches a \localPAV{$\frac{1}{2k^2}$} committee.)
The final committee $W$ satisfies \EJRplus{} alongside the aforementioned efficiency guarantees.

\paragraph{Existence of \JR{} and \fPO{} Committees.}
The literature on discrete fair division has developed several efficient algorithms for finding fair and Pareto optimal allocations~\citep{BKV18ef1,GM24fPO,M25polynomial}. These algorithms predominantly produce \fPO{} allocations, as this class is more structured and can be verified tractably. In contrast, for approval-based committee voting, no efficient algorithms are known for finding committees that are \EJRplus{} (or even just \JR{}) and \PO{}. A natural approach to designing such an algorithm is to investigate the \emph{existence} of committees satisfying \EJRplus{} (or \JR{}) alongside \fPO{}. 

Shortly before publishing this paper as a preprint, \citet[Theorem 1]{jin2026limitationsbobw} discovered a sophisticated counterexample where no \JR{} committee satisfies \fPO{}.
While we were working on this paper, this existence question had not yet been solved, and we had tried many natural strategies to prove existence. In \Cref{appendix:Limitations}, we provide counterexamples for each. While these are now formally superseded, we include them in case they prove useful in future research (perhaps about relaxations) as test cases. Our attempted proof strategies include maximizing PAV score subject to \fPO{}, maximizing approval score subject to \JR{}, defining analogous algorithms to the ones used in fair allocation and fractional committee voting, looking for a committee that maximizes weighted utilitarian welfare where the weights satisfy certain sufficient conditions to guarantee \JR{}, as well as rounding efficient fractional committees.

\subsection*{Related Work}

The computational complexity of PAV, as well as the broader class of Thiele methods (which replace the harmonic numbers $1 + \frac12 + \dots + \frac{1}{|W \cap A_i|}$ in the objective by other values $w_1 + w_2 + \dots + w_{|W \cap A_i|}$) has been widely studied. Computing an optimal PAV committee is known to be \NPH{}~\citep[Corollary 2]{AGG+15computationalaspects}, and the same is true for all Thiele methods for which there exists some $t < k$ with $w_t > w_{t+1}$ \citep[Theorem 5]{SFL16finding}. On some restricted domains, the problem becomes polynomial-time solvable, including the candidate interval (CI) domain \citep[Corollary 14]{PL20spoc} and voter interval (VI) domain \citep{ALS+26computing,MS26polynomial}; however, the problem remains hard on the party approval domain where each candidate has $k$ copies \citep[Theorem 5.1]{BGPSW24approvalapportionment}. %
The problem is fixed-parameter tractable (\textrm{\textup{FPT}}) with respect to several parameters \citep{YW23parameterized}.

Pareto optimality in committee elections was studied from a computational point of view by \citet[Section 6.1]{AM20computing} who proved that checking whether a given committee is \PO{} is \coNPC{} and \textrm{\textup{W[2]-complete}} with respect to parameter $k$. \citet[Section 6]{LS20utilitarian} considered \PO{} as an axiom for committee voting rules, showing that PAV satisfies it but that a variety of other rules fail it (including Sequential Phragm\'en, Monroe's rule, and Sequential PAV); see \citet[Appendix A.1]{LS23book} for further examples, including the Method of Equal Shares. \citet{S26pareto} studied \PO{} in restricted domains including CI and VI, providing an efficient algorithm for finding a committee satisfying \EJRplus{} and \PO{} in these domains. He also studied the connection between \PO{} and committee monotonicity and reconfiguration. \citet{BBF+26fractional} study \fPO{} for approval-based committee elections, providing equivalent definitions and verification algorithms and showing that PAV violates \fPO{}.

At a technical level, the work of \citetalias{DMM+21tight} is closest to ours. They present a $0.79$-approximation algorithm for the PAV score by optimizing a piecewise-linear relaxation and applying pipage rounding, and show that obtaining a better approximation is \NPH{}. Their analysis applies more generally to \emph{geometrically dominant} Thiele rules, including approval voting, Chamberlain-Courant, and PAV. These rules are defined by a nonnegative, nonincreasing weight vector $\mathbf{w}$, normalized by $w_1 = 1$, that satisfies $w_i \cdot w_{i+2} \geq w_{i+1}^2$ for all $i \in \mathbb{N}$. The authors provide a $\gamma_{\w}$-approximation algorithm for any geometrically dominant Thiele rule $\mathbf{w}$, where $\gamma_{\w} \coloneqq \E_{t \sim \Poi(1)}[w_1 + \dots + w_t]$ and $\Poi(1)$ denotes a Poisson random variable with mean $1$. Additionally, if $\lim_{i\to\infty} w_i = 0$, they show that for any $\eps > 0$, it is NP-hard to compute a solution that is $(\gamma_{\w} + \eps)$-approximate. The $\gamma_{\w}$ value for Chamberlain-Courant is $1-1/e \approx 0.632$, and that for PAV is $\approx 0.796$. We will use $\gamma^\star$ to denote the inapproximability threshold for PAV.

However, their algorithm, which we will refer to as the DMMS algorithm, differs from ours in two important ways. The first difference is the choice of relaxation, plotted in~\Cref{fig:harmonic}. Their algorithm optimizes a piecewise linear interpolation of the PAV function, while we use a different objective explained in the next paragraph. Second, their algorithm does not involve a local search phase, and therefore its output committee can fail to satisfy \JR{}~(\Cref{appendix:DMMS_Fails_JR}). Importantly, augmenting their algorithm with local search can ensure justified representation but the output might not satisfy $\alpha^\star$-\fPO{} (see~\Cref{lem:dmms_not_JR}). Our proof of $\alpha^\star$-\fPO{} (briefly sketched in the next three paragraphs) requires a concave surrogate that lower bounds the multilinear extension of PAV -- a property that their piecewise linear objective does not satisfy.

Our algorithm optimizes the objective $ \Psi(x)=\sum_i h(u_i(x))$, where $h(z)=\int_0^z \frac{1-e^{-y}}{y}\,dy$. The $\Psi(\cdot)$ function serves as a lower surrogate for the multilinear extension of PAV, i.e., $F_{\PAV}(x) \geq \Psi(x)$ for every fractional committee $x$. Furthermore, if $\Psi(x) \geq 1/4$, we can show a stronger bound $F_{\PAV}(x) \geq \Psi(x) + 2\eta$, with $\eta = \frac{1}{1024n^2mk^2}$. To avoid real-arithmetic issues with computing an exact optimum, our algorithm finds an approximate $\Psi$-optimum $\widehat{x}$; specifically, $\Psi(\widehat{x}) \geq \Psi^\star - \eta$ where $\Psi^\star = \max_x \Psi(x)$. It can be shown that, for any instance where at least one voter approves at least one candidate, the optimum objective value is $\Psi^\star > 3/4$, and by the choice of $\eta$, $\Psi(\widehat{x}) \geq 1/4$. Hence, the slack of $2\eta$ holds without loss of generality.

After finding the approximate $\Psi$-optimum $\widehat{x}$, pipage rounding with respect to $F_{\PAV}$ 
returns an integral committee $W_0$ with
$$\PAV(W_0) \geq F_{\PAV}(\widehat{x}) \geq \Psi(\widehat{x}) + 2\eta \geq \Psi^\star + \eta > \Psi^\star,$$
and subsequent local search steps only increase the PAV score. Hence, the final committee $W$ satisfies $\PAV(W) > \Psi^\star$.

If $W$ were fractionally $\alpha^\star$-Pareto dominated by some fractional committee $y$, then the inequality $h(\alpha^\star t)\geq H_t$ for all integers $t\geq 0$ would imply $\Psi(y)>\PAV(W)$, contradicting $\PAV(W) > \Psi^\star$. 
Thus, the central distinction is that our proof uses a lower surrogate objective that is preserved without loss by rounding, whereas DMMS optimizes a linearized relaxation whose value is preserved only up to the approximation factor $\gamma^\star$.

Finally, the work of \citet{BFF21tight} is relevant to ours. They study approximation algorithms for concave coverage objectives, of which PAV is a special case. The pointwise inequality $F_{\PAV}(x)\geq\Psi(x)$ used in our analysis is the PAV specialization of their Bernoulli--Poisson inequality~(see Lemma 2.3 in their paper).

\section{Preliminaries}
\label{sec:Preliminaries}

Let $\N \coloneqq \{1,2,\dots\}$. Given any $r \in \mathbb{N}$, let $[r] \coloneqq \{1,2,\dots,r\}$.

\paragraph{Problem instance.} An \emph{instance} $\I$ of the approval-based committee voting problem is given by a tuple $\I = \langle N, C, \A, k \rangle$, where $N \coloneqq [n]$ is the set of $n$ \emph{voters}, $C \coloneqq \{c_1,c_2,\dots,c_m\}$ is the set of $m$ candidates, $\A = (A_1,\dots,A_n)$ is the \emph{preference profile}, where $A_i \subseteq C$ denotes the set of candidates approved by voter $i \in N$, and $2 \le k \le m - 1$ is the \emph{committee size}. We will assume that each voter approves at least one candidate; however, there might be a candidate that no voter approves.

\paragraph{Integral and fractional committees.} A feasible integral committee for the instance $\I = \langle N, C, \A, k \rangle$, denoted by $W \subseteq C$, is a subset of exactly $k$ candidates. A feasible fractional committee, denoted by $x \in [0,1]^C$, is an entrywise nonnegative vector of length $m$ whose entries sum to exactly $k$, i.e., $\sum_{c\in C} x_c = k$.
We will use the term ``committee'' to refer to a feasible integral committee, and ``fractional committee'' to denote a feasible fractional committee. 

\paragraph{Utility of a voter.} Given an instance $\I = \langle N, C, \A, k \rangle$, the utility derived by voter $i \in N$ from a committee $W$ is defined as the number of candidates in $W$ approved by voter $i$. That is, 
\[u_i(W) \coloneqq |W \cap A_i|.\]
For a fractional committee $x\in[0,1]^C$, define $u_i(x) \coloneqq \sum_{c \in A_i} x_c$. A voter $i \in N$ is said to be \emph{represented} by a committee $W$ if $u_i(W) > 0$ (equivalently, $W \cap A_i \neq \emptyset$); otherwise, the voter is deemed \emph{unrepresented}. The \emph{support} of a candidate $c \in C$ is defined as the set of voters who approve it, i.e., $\supp(c) \coloneqq \{i \in N : c \in A_i\}$.

\paragraph{Pareto optimality (\PO{}).} 
For $\alpha \geq 1$, an integral committee $W'$ is said to $\alpha$-dominate another integral committee $W$ if $u_i(W') \geq \alpha \cdot u_i(W)$ for every voter $i \in N$ and the inequality is strict for at least one voter. A committee $W$ is \emph{$\alpha$-Pareto optimal} ($\alpha$-\PO{}) if no other integral committee $\alpha$-dominates it~\citep{PY00approximability,ILW+17approximate}. It is \emph{Pareto optimal} if it is $1$-\PO{}. The committee $W$ is \emph{weakly Pareto optimal} (weak PO) if no other integral committee $W'$ strictly dominates it in the sense that $u_i(W') > u_i(W)$ for all $i \in N$.

\paragraph{Fractional Pareto optimality (\fPO{}).} 
For $\alpha \geq 1$, a fractional committee $x$ is said to $\alpha$-dominate an integral committee $W$ if $u_i(x) \geq \alpha \cdot u_i(W)$ for every voter $i \in N$ and the inequality is strict for at least one voter. A committee $W$ is \emph{$\alpha$-fractionally Pareto optimal} ($\alpha$-\fPO{}) if no fractional committee $\alpha$-dominates it. It is \emph{fractionally Pareto optimal} (\fPO{}) if it is 1-\fPO{}.
Note that for any $\alpha \geq 1$, $\alpha$-\fPO{} implies $\alpha$-\PO{}.

\paragraph{Representation notions.} A committee $W$ is said to satisfy
\begin{itemize}
    \item \emph{justified representation} (\JR{})~\citep{ABC+17justified} if no group of at least $n/k$ unrepresented voters commonly approves some candidate. That is, for every candidate $c\in C$, we have that $\left| \{i \in \supp(c) : u_i(W)=0\} \right| < \frac{n}{k}$.
    \item \emph{extended justified representation} (\EJR{})~\citep{ABC+17justified} if for every $\ell \in [k]$ and every group $S \subseteq N$ with $|S| \ge \ell \cdot \frac{n}{k}$ and $|\bigcap_{i \in S} A_i| \ge \ell$, there exists a voter $i \in S$ with $u_i(W) \ge \ell$.
    \item \EJRplus{}~\citep{BP23robust} if for every $\ell \in [k]$ and every group $S \subseteq N$ with $|S| \ge \ell \cdot \frac{n}{k}$ such that there is some $c \in C \setminus W$ with $c \in \bigcap_{i \in S} A_i$, there exists a voter $i \in S$ with $u_i(W) \ge \ell$.
\end{itemize}

Observe that $\EJRplus{} \Rightarrow \EJR{} \Rightarrow \JR{}$ and these implications are strict.

\paragraph{Voting rules.} An approval-based committee voting rule takes an election instance $\I = \langle N, C, \A, k \rangle$ as input and outputs a feasible and integral winning committee $W$ for the instance $\I$. Below, we define some of the voting rules considered in our study:
\begin{itemize}
    \item \emph{Thiele Methods}: A Thiele method is an approval-based committee voting rule specified by a weight sequence $\mathbf{w}=(w_1,w_2,\dots)$, where $w_j \geq 0$ for every $j \geq 1$. The $\mathbf{w}$-Thiele score of a committee $W \subseteq C$ is $\operatorname{score}_{\mathbf{w}}(W) \coloneqq \sum_{i \in N} \sum_{j=1}^{u_i(W)} w_j$. The corresponding Thiele rule returns a size-$k$ committee that maximizes this score. Observe that every Thiele committee is \PO{} when $w_j > 0$ for every $j \ge 1$.
    \item \emph{Approval Voting} (AV): The approval voting rule is a Thiele method with the all-ones weight sequence $\mathbf{w} = (1,1,\dots,1)$. Equivalently, any committee $W$ chosen by the approval voting rule maximizes the utilitarian social welfare $\sum_{i \in N} u_i(W)$. It is known that an AV committee can fail \JR{}~\citep{ABC+17justified}. However, due to the utilitarian welfare maximization property, any AV committee satisfies \fPO{}.
    \item \emph{Proportional Approval Voting} (PAV): For $t \in \N$, denote the $t^\textup{th}$ harmonic number by 
    \[H_t \coloneqq 1 + \frac{1}{2} + \cdots + \frac{1}{t},\]
    and let $H_0 \coloneqq 0$. The PAV score of a committee $W$ is defined as 
    \[\PAV(W) \coloneqq \sum_{i \in N} H_{u_i(W)}.\]
    The PAV rule returns a committee maximizing the PAV score. In other words, PAV corresponds to a Thiele method with weights $w_j = 1/j$ for all $j \in [m]$. 
    Every PAV committee satisfies \EJRplus{}~\citep{BP23robust}. However, for any constant $\eps > 0$, it is \NPH{} to compute a committee that approximates the optimal PAV score to within a factor $(\gamma^\star + \eps)$, 
    where $\gamma^\star \coloneqq \E_{t \sim \Poi(1)}[H_t] \approx 0.79$~\citep{DMM+21tight}. In other words, $\gamma^\star$ is the expected contribution of a voter to the PAV score when the number of its approved candidates chosen in the committee is distributed as a Poisson random variable with mean $1$.
    \item \emph{$\tau$-Local Search PAV rule} (\localPAV{$\tau$}): Given any $\tau > 0$, a committee $W$ is said to be a $\tau$-\emph{local PAV committee} if its PAV score cannot be increased by at least $\tau$ by replacing a single candidate in $W$ with a candidate outside of $W$. That is, for all $c \in W$ and $c' \in C \setminus W$,
    \[\PAV(W \setminus \{c\} \cup \{c'\}) - \PAV(W) < \tau.\]

    A \localPAV{$\tau$} committee can be computed by the natural local search procedure in polynomial time whenever $\tau$ is at least an inverse polynomial in the input size, i.e., $\tau \geq \frac{1}{\textup{poly}{(n,k)}}$. For exact local PAV, however, \citet{KE24lower}, construst a lexicographically chosen better-response sequence of length $\Omega(k^{\log k})$; thus, in particular, no polynomial bound is known for exact local search. Additionally, every \localPAV{$\tau$} committee satisfies \EJRplus{} if 
    
    $\tau \leq \frac{n}{k^2}$~\citep{AEH+18complexity,BP23robust}. 
    
    For $\tau=0$, we use the convention that \localPAV{$0$} means \emph{exact local PAV}, i.e., no swap strictly improves the PAV score. Note that for any $\tau < \tau'$, any \localPAV{$\tau$} committee is also a \localPAV{$\tau'$} committee. Thus, for any instance and any $\tau > 0$, if a property is satisfied by all \localPAV{$\tau$} committees, then the same is also satisfied by all exact local PAV committees.
\end{itemize}

\paragraph{PAV score approximation.} Given an instance $\I$, let $\OPT(\I)$ denote the highest possible PAV score of any committee in $\I$. Given any $\gamma \geq 0$, a committee $W$ is said to provide a $\gamma$-approximation to the PAV score if $\PAV(W) \geq \gamma \OPT(\I)$. It is known that unless $\textrm{\textup{P}} = \NP{}$, no polynomial-time algorithm can provide better than a $\gamma^\star \approx 0.79$ approximation to the PAV score~\citep{DMM+21tight}.

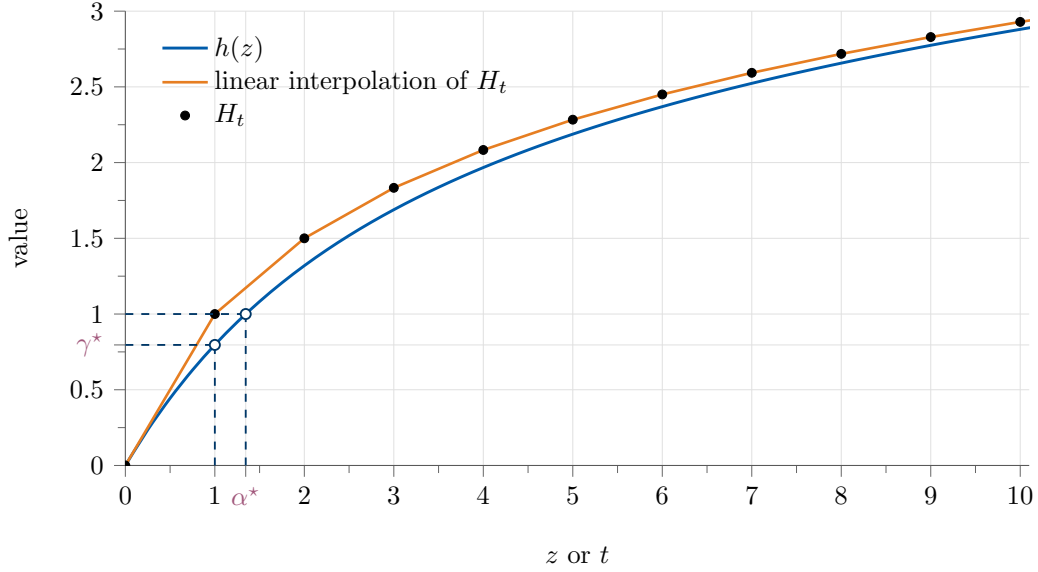
\begin{figure}[t]
    \centering
    \definecolor{harmonicblue}{RGB}{0,92,171}
\definecolor{harmonicorange}{RGB}{230,126,34}
\pgfmathsetmacro{\alphastar}{1.345016617}
\pgfmathsetmacro{\gammastar}{0.796600}

\begin{tikzpicture}
\begin{axis}[
    width=0.82\textwidth,
    height=0.46\textwidth,
    xmin=0, xmax=10.1,
    ymin=0, ymax=3,
    axis x line*=bottom,
    axis y line*=left,
    axis line style={black!65},
    enlargelimits=false,
    xlabel={$z$ or $t$},
    ylabel={value},
    xtick={0,1,...,10},
    ytick={0,0.5,...,3},
    extra x ticks={\alphastar},
    extra x tick labels={$\alpha^\star$},
    extra x tick style={
        tick label style={
            text=magenta!60!black,
        },
    },
    extra y ticks={\gammastar},
    extra y tick labels={$\gamma^\star$},
    extra y tick style={
        tick label style={
            text=magenta!60!black,
        },
    },
    minor tick num=1,
    tick align=outside,
    tick style={black!60},
    major grid style={draw=black!12, line width=0.25pt},
    grid=major,
    clip=true,
    clip marker paths=true,
    label style={font=\small},
    tick label style={font=\small},
    legend style={
        draw=none,
        fill=none,
        font=\small,
        at={(0.03,0.97)},
        anchor=north west,
        cells={anchor=west},
    },
]

\pgfmathsetmacro{\alphastar}{1.345}
\pgfmathsetmacro{\gammastar}{0.796}

\addplot[
    harmonicblue,
    line width=1.1pt,
    smooth,
    mark=none,
] coordinates {
    (0.00,0.000000)
    (0.25,0.235204)
    (0.50,0.443842)
    (0.75,0.629874)
    (1.00,0.796600)
    (1.25,0.946773)
    (1.50,1.082700)
    (1.75,1.206320)
    (2.00,1.319263)
    (2.25,1.422908)
    (2.50,1.518421)
    (2.75,1.606796)
    (3.00,1.688876)
    (3.25,1.765387)
    (3.50,1.836949)
    (3.75,1.904096)
    (4.00,1.967289)
    (4.25,2.026930)
    (4.50,2.083366)
    (4.75,2.136902)
    (5.00,2.187802)
    (5.25,2.236301)
    (5.50,2.282605)
    (5.75,2.326896)
    (6.00,2.369335)
    (6.25,2.410068)
    (6.50,2.449221)
    (6.75,2.486911)
    (7.00,2.523241)
    (7.25,2.558304)
    (7.50,2.592185)
    (7.75,2.624958)
    (8.00,2.656695)
    (8.25,2.687457)
    (8.50,2.717303)
    (8.75,2.746286)
    (9.00,2.774453)
    (9.25,2.801849)
    (9.50,2.828515)
    (9.75,2.854488)
    (10.00,2.879805)
    (10.25,2.904497)
    (10.50,2.928593)
    (10.75,2.952123)
    (11.00,2.975112)
    (11.25,2.997585)
    (11.50,3.019564)
    (11.75,3.041070)
    (12.00,3.062123)
    (12.25,3.082742)
    (12.50,3.102945)
    (12.75,3.122747)
    (13.00,3.142165)
    (13.25,3.161213)
    (13.50,3.179905)
    (13.75,3.198255)
    (14.00,3.216273)
    (14.25,3.233973)
    (14.50,3.251364)
    (14.75,3.268459)
    (15.00,3.285266)
    (15.25,3.301795)
    (15.50,3.318056)
    (15.75,3.334056)
    (16.00,3.349804)
    (16.25,3.365309)
    (16.50,3.380576)
    (16.75,3.395614)
    (17.00,3.410429)
    (17.25,3.425028)
    (17.50,3.439417)
    (17.75,3.453601)
    (18.00,3.467587)
    (18.25,3.481381)
    (18.50,3.494986)
    (18.75,3.508409)
    (19.00,3.521655)
    (19.25,3.534727)
    (19.50,3.547630)
    (19.75,3.560369)
    (20.00,3.572948)
};
\addlegendentry{$h(z)$}

\addplot[
    harmonicorange,
    line width=1.0pt,
    mark=none,
] coordinates {
    (0,0.000000)
    (1,1.000000)
    (2,1.500000)
    (3,1.833333)
    (4,2.083333)
    (5,2.283333)
    (6,2.450000)
    (7,2.592857)
    (8,2.717857)
    (9,2.828968)
    (10,2.928968)
    (11,3.019877)
    (12,3.103211)
    (13,3.180134)
    (14,3.251562)
    (15,3.318229)
    (16,3.380729)
    (17,3.439553)
    (18,3.495108)
    (19,3.547740)
    (20,3.597740)
};
\addlegendentry{linear interpolation of $H_t$}

\addplot[
    only marks,
    mark=*,
    mark size=1.65pt,
    mark options={draw=black, fill=black},
] coordinates {
    (0,0.000000)
    (1,1.000000)
    (2,1.500000)
    (3,1.833333)
    (4,2.083333)
    (5,2.283333)
    (6,2.450000)
    (7,2.592857)
    (8,2.717857)
    (9,2.828968)
    (10,2.928968)
    (11,3.019877)
    (12,3.103211)
    (13,3.180134)
    (14,3.251562)
    (15,3.318229)
    (16,3.380729)
    (17,3.439553)
    (18,3.495108)
    (19,3.547740)
    (20,3.597740)
};
\addlegendentry{$H_t$}

\addplot[
    harmonicblue!60!black,
    dashed,
    line width=0.7pt,
    forget plot,
] coordinates {(\alphastar,0) (\alphastar,1)};
\addplot[
    harmonicblue!60!black,
    dashed,
    line width=0.7pt,
    forget plot,
] coordinates {(1,0) (1,\gammastar)};
\addplot[
    harmonicblue!60!black,
    dashed,
    line width=0.7pt,
    forget plot,
] coordinates {(0,\gammastar) (1,\gammastar)};
\addplot[
    harmonicblue!60!black,
    dashed,
    line width=0.7pt,
    forget plot,
] coordinates {(0,1) (\alphastar,1)};
\addplot[
    only marks,
    mark=*,
    mark size=1.9pt,
    mark options={draw=harmonicblue!60!black, fill=white, line width=0.7pt},
    forget plot,
] coordinates {(\alphastar,1) (1,\gammastar)};
\end{axis}
\end{tikzpicture}
    \caption{Comparison of $h(z)$ with the harmonic numbers $H_t$ and their piecewise linear interpolation. We show the points $\alpha^\star \approx 1.346$, where $h(\alpha^\star) = 1$, and $\gamma^\star = h(1) \approx 0.79$.}
    \label{fig:harmonic}
\end{figure}

\paragraph{The function $h(\cdot)$ and the constant $\alpha^\star$.} 
Define the function
\[h(z) \coloneqq \int_0^z \frac{1-e^{-y}}{y} \, dy, \qquad z \geq 0,\]
with the integrand interpreted as $1$ at $y=0$. Observe that $h(\cdot)$ is strictly increasing and concave.\footnote{Indeed, the derivative of $h(\cdot)$ is $h'(z) = \frac{1-e^{-z}}{z}$, with $h'(0)=1$. Also, $h''(z) = \frac{(z+1)e^{-z}-1}{z^2} \leq 0$ because $(z+1) e^{-z} \leq 1$. Thus, $h$ is increasing and concave.} 

The function $h(z)$ is also known as the \emph{entire exponential integral}, usually denoted by $\operatorname{Ein}(z)$; see, for example, the DLMF entry on exponential integrals~\citep{DLMF}. In the multiwinner voting literature, this function appears implicitly in the analysis of \citet{DMM+21tight}.
Their approximation factor for a Thiele weight sequence $\mathbf{w}=(w_1,w_2,\ldots)$ is expressed through the expected contribution of a voter to the Thiele score when the number of approved elected candidates is Poisson-distributed. As we show in \Cref{lem:Poisson-smoothening} (see \Cref{appendix:h_and_harmonics} for a proof), $h(\cdot)$ is a Poisson-smoothed version of the harmonic utility function. 

\begin{restatable}[]{lemma}{PoissonSmoothening}
For every $\mu \geq 0$, if $Z$ is a Poisson random variable with mean $\mu$, then $h(\mu) = \E[H_Z]$.

\label{lem:Poisson-smoothening}
\end{restatable}

For every positive integer $t$, the $t^\textup{th}$ harmonic number $H_t$ is lower bounded by $h(t)$ and upper bounded by $\frac{1}{\gamma^\star} h(t)$, where $\gamma^\star \coloneqq h(1) \approx 0.79$~(\Cref{lem:h_Lower_and_Upper_Bound_Harmonics}; see \Cref{appendix:h_and_harmonics} for a proof). \Cref{fig:harmonic} illustrates how $h(\cdot)$ lower bounds the piecewise-linear interpolation of harmonic numbers.

\begin{restatable}[]{lemma}{hLowerAndUpperBoundHarmonics}
For every integer $t \geq 1$, 
\[\gamma^\star H_t \le h(t) < H_t,\]
where $\gamma^\star \coloneqq h(1)=\int_0^1\frac{1-e^{-y}}{y}\,dy
  \approx 0.79$.
\label{lem:h_Lower_and_Upper_Bound_Harmonics}
\end{restatable}

Let $\alpha^\star$ be the unique positive solution of $h(\alpha^\star)=1$. Numerically, $\alpha^\star = 1.345016617\ldots$. The following inequality relating a scaled version of the function $h(\cdot)$ to the harmonic function $H_t$ will be useful in our analysis. Its proof can be found in \Cref{appendix:h_and_harmonics} in the appendix.

\begin{restatable}[]{lemma}{MultiplicativeShiftLemma}
For every integer $t \geq 0$, 
\[h(\alpha^\star t)\ge H_t,\]
where $\alpha^\star \approx 1.346$ is the unique positive solution of $h(\alpha^\star)=1$.
\label{lem:Multiplicative-Shift-ineq}
\end{restatable}

\paragraph{Submodular function.}
Let $\Omega = \{\omega_1,\dots,\omega_m\}$ be a finite set. A set function $f: 2^\Omega \to \mathbb{R}$ is said to be monotone non-decreasing (or simply \emph{monotone}) if for every pair of subsets $S,T \subseteq \Omega$ such that $S \subseteq T$, we have $f(S) \leq f(T)$. The function $f$ is said to be \emph{submodular} if for every pair of subsets $S,T \subseteq \Omega$ such that $S \subseteq T$ and any $\omega \notin T$, we have $f(S \cup \{\omega\}) - f(S) \geq f(T \cup \{\omega\}) - f(T)$. That is, a submodular function satisfies the ``diminishing marginal returns'' property.

\paragraph{PAV as a submodular function and pipage rounding.}

The set function $\PAV:2^C \to \mathbb{R}_{\geq 0}$ is monotone and submodular:
Indeed, for each voter $i$, the marginal contribution of adding a candidate
approved by $i$ is $1/(u_i(W)+1)$, which weakly decreases as $W$ grows.
Its \emph{multilinear extension} is the function $F_{\PAV}:[0,1]^C \to \mathbb{R}_{\geq 0}$
defined by
\[
    F_{\PAV}(x)
    \coloneqq
    \mathbb{E}[\PAV(R(x))]
    =
    \sum_{i \in N} \mathbb{E}\!\left[H_{u_i(R(x))}\right],
\]
where the random set $R(x) \subseteq C$ includes each candidate $c \in C$ independently with probability $x_c$. 

A useful property of $F_{\PAV}$ is that, for every pair of distinct elements $i,j \in C$, the univariate function $z \mapsto F_{\PAV}(x+z(\mathbf{1}_i-\mathbf{1}_j))$ is convex on its feasible interval; this follows from submodularity of PAV and multilinearity of $F_{\PAV}$. This convexity is the basis of \emph{pipage rounding}, introduced by \citet{AS04pipage} and extended to the multilinear-extension framework for matroid constraints by \citet{CCP+11maximizing}. \Cref{prop:Pipage_for_PAV} below adapts the standard pipage rounding guarantee to the PAV setting. It is stated in terms of the fractional relaxation of the committee-size constraint, denoted by $P_k \coloneqq \left\{x \in [0,1]^C : \sum_{c \in C} x_c = k \right\}$, which is the base polytope of the rank-$k$ uniform matroid on $C$.

\begin{proposition}[Pipage rounding for PAV;~\citealp{CCP+11maximizing}] 
\label{prop:Pipage_for_PAV}
There is a polynomial-time randomized algorithm that, given any fractional solution $x \in P_k$, returns a random committee $W \subseteq C$ of size $k$ such that 
\[\mathbb{E}[\PAV(W)] \geq F_{\PAV}(x).\]
Moreover, the algorithm preserves marginals, i.e., for every candidate $c \in C$, $\Pr[c \in W] = x_c.$ In addition, since $F_{\PAV}$ can be evaluated in polynomial time, this rounding can be derandomized. That is, there is a polynomial-time deterministic algorithm that, given any $x \in P_k$, returns a committee $W \subseteq C$ of size $k$ such that 
\[\PAV(W) \geq F_{\PAV}(x).\]
\end{proposition}

\section{Results}
\label{sec:Results}

We will now present our positive results, starting with the analysis of the global and local versions of the PAV rule (\Cref{subsec:Results_local_PAV,subsec:Results_global_PAV}), followed by our round-and-swap algorithm~(\Cref{alg:Pipage-and-swap}).

\subsection{Results for \localPAV{$\tau$}}
\label{subsec:Results_local_PAV}

Let us start by analyzing the \localPAV{$\tau$} family of voting rules. For any fixed committee size $k$ and number of voters $n$, we say that a property is \emph{guaranteed} by the \localPAV{$\tau$} rule if \emph{every} \localPAV{$\tau$} committee satisfies that property. Conversely, if some \localPAV{$\tau$} committee does not satisfy some property, we say that the \localPAV{$\tau$} rule fails that property.

\Cref{tab:Local-PAV-Results} presents the range of values for $\tau$ for which \localPAV{$\tau$} satisfies various binary properties, such as \PO{}, weak \PO{}, \EJRplus{},%
\footnote{For \EJRplus{}, it is a classic result that $\tau \le n/k^2$ suffices to guarantee it \citep{AEH+18complexity}.}
$(1+\eps)$-\PO{} for all $\eps>0$, $2$-\fPO{}, and a $2/3$-approximation to the PAV score. The last two approximations represent the best possible guarantees achievable by any \localPAV{$\tau$} committee.

\begin{table}[ht]
\renewcommand{\arraystretch}{1.35}
\begin{tabularx}{\textwidth}{>{\raggedright\arraybackslash}p{0.17\textwidth} >{\raggedright\arraybackslash}p{0.29\textwidth} >{\raggedright\arraybackslash}X}
\toprule
Property & Guaranteed range for every \localPAV{$\tau$} committee & Counterexamples / tightness comments \\
\midrule
\PO{} & Fails for all $\tau \geq 0$. & Exact local PAV can return a Pareto-dominated committee. \\
\midrule[0.01pt]
Weak \PO{} & $\tau\le n/k^2$. & Tight: For every $k\ge 2$, there is an instance with a non-weakly-\PO{} $\tau$-local committee for all $\tau>n/k^2$. \\
\midrule[0.01pt]
\EJRplus{} & $\tau\le n/k^2$. & Tight, using the same family as for weak \PO{}. \\
\midrule[0.01pt]
$(1+\varepsilon)$-\PO{} for all $\varepsilon>0$ & $\tau\le 1/k^2$. & Counterexamples exist for every $\tau>1/k$. 
\\
\midrule[0.01pt]
$2$-\fPO{} & $0 \leq \tau \leq \frac{1}{2k^2}$ & The factor $2$ is asymptotically tight already for exact local PAV. The property is not guaranteed whenever $\tau>1/\lceil k/2\rceil$.\\
\midrule[0.01pt]
$2/3$ PAV-score approximation & $0\le\tau\le \frac{1}{2k^2}$ & The factor $2/3$ is asymptotically tight already for exact local PAV. The property is not guaranteed whenever $\tau>1/(\sigma_k+1)$. \\
\bottomrule
\end{tabularx}
\caption{Results for \localPAV{$\tau$} rule. Positive results pertain to satisfaction of a property by all \localPAV{$\tau$} committees, while counterexamples imply that the property is violated by some \localPAV{$\tau$} committee. Here, $\sigma_k \coloneqq \max\left\{s\in\{0,1,\dots,k-1\}: H_s < \frac{2}{3}H_k\right\} = \Theta(k^{2/3})$.
}
\label{tab:Local-PAV-Results}
\end{table}

\begin{restatable}[Properties of local PAV]{theorem}{LocalPAVProperties}
Fix an approval-based multiwinner instance with $n$ voters and committee size $k$. The instance-wise guarantees and counterexamples in \Cref{tab:Local-PAV-Results} hold.
\label{thm:local_PAV_properties}
\end{restatable}

Towards establishing \Cref{thm:local_PAV_properties}, we will present the proofs for weak Pareto optimality, $2$-\fPO{}, and $2/3$-PAV guarantees in this section. The remaining claims and tightness results are proved in \Cref{appendix:Proof_local_PAV_properties}.

\medskip
\noindent\emph{Weak Pareto optimality.}
Let $W$ be a $\tau$-local PAV committee with $\tau\le n/k^2$. Suppose, for contradiction, that $W$ is not weakly Pareto optimal. Then there exists an integral committee $W'$ of size $k$ such that
$u_i(W') > u_i(W)$ for every voter $i$. 

Let $X \coloneqq W'\setminus W$, $Y \coloneqq W\setminus W'$, and $r \coloneqq |X|=|Y|$. Since $W' \neq W$, we have $r \geq 1$. For every $y \in Y$ and $x \in X$, define the change in the PAV score upon swapping $y$ and $x$ as
\[\Delta_{y,x} \coloneqq \PAV(W \setminus \{y\} \cup \{x\})-\PAV(W).\]
Since $W$ is a \localPAV{$\tau$} committee, we have
$$\sum_{y\in Y}\sum_{x\in X}\Delta_{y,x}<r^2\tau.$$
We will now lower-bound the double summation. Fix a voter $i$, and write
\[u_i\coloneqq |A_i\cap W|, \qquad x_i\coloneqq |X\cap A_i|, \qquad y_i\coloneqq |Y\cap A_i|.\]
Then, $u_i(W') = u_i - y_i + x_i$. Since $W'$ strongly improves voter $i$ and utilities are integral, we have that $x_i \geq y_i + 1$.

Let $\Gamma_i$ denote voter $i$'s total contribution to $\sum_{y\in Y}\sum_{x\in X}\Delta_{y,x}$. If $u_i=0$, then $y_i=0$, and
$$\Gamma_i=rx_i\ge r\ge \frac{r}{k}.$$
Now suppose $u_i\ge 1$. The voter contributes a PAV score of $1/(u_i+1)$ for swaps where the added candidate is approved and the removed candidate is not approved, and reduces the PAV score by $1/u_i$ for swaps where the removed candidate is approved and the added candidate is not approved. Hence,
\[\Gamma_i = \frac{x_i (r - y_i)}{u_i + 1} - \frac{y_i(r - x_i)}{u_i}.\]
This expression is increasing in $x_i$, so it is minimized when $x_i=y_i+1$. That is,
\[\Gamma_i \geq \frac{(y_i + 1) (r - y_i)}{u_i + 1} - \frac{y_i(r - y_i - 1)}{u_i}.\]
Substituting $u_i = y_i + z_i$, where $z_i \coloneqq u_i - y_i$, gives
\[\Gamma_i \geq \frac{rz_i + y_i(y_i + 1)}{(y_i + z_i)(y_i + z_i + 1)}.\]
Note that the denominator is positive since $y_i + z_i = u_i \geq 1$. 

We will now show that the right-hand side is at least $r/k$, i.e., $k \bigl( rz_i + y_i(y_i + 1) \bigr) - r(y_i + z_i)(y_i + z_i + 1) \geq 0$. To show this, define
\[a \coloneqq k - r - z_i,
\qquad b \coloneqq r - y_i - 1.\]
Observe that $a \geq 0$ since $z_i = u_i - y_i \leq |W\cap W'| = k - r$, and $b \geq 0$ since $r = |X| \geq x_i \geq y_i + 1$. A direct expansion gives
\[
\begin{aligned}
&k\bigl(rz_i+y_i(y_i+1)\bigr)-r(y_i+z_i)(y_i+z_i+1) \\
&\qquad =
 a b z_i+a z_i y_i+a z_i+a y_i^2+a y_i+b^2 z_i+bz_i\ge 0.
\end{aligned}
\]
Therefore, $\Gamma_i\ge r/k$ for every voter. Summing over voters,
\[
\sum_{y\in Y}\sum_{x\in X}\Delta_{y,x}
=
\sum_{i\in N}\Gamma_i
\ge n\frac{r}{k}.
\]
Since $r\le k$, we have $nr/k\ge nr^2/k^2$. Hence,
\[
n\frac{r^2}{k^2}
\le
n\frac{r}{k}
\le
\sum_{y\in Y}\sum_{x\in X}\Delta_{y,x}
<
r^2\tau
\le
r^2\frac{n}{k^2},
\]
a contradiction. Thus, $W$ is weakly Pareto optimal.

\par\medskip
\paragraph{$2$-\fPO{} guarantee.}
We will prove that every $\tau$-local PAV committee is $2$-\fPO{} whenever $0 \leq \tau \le \frac{1}{2k^2}$.

Let $W$ be a $\tau$-local PAV committee. Suppose, for contradiction, that $W$ fails $2$-\fPO{}. Then, there must exist a fractional committee $x$ such that $u_i(x)\ge 2u_i(W)$ for every voter $i$ with strict inequality for at least one voter. For any candidate $a\in W$, define
$\lambda_a\coloneqq 1-x_a$, and for any candidate $b \notin W$, define $\mu_b\coloneqq x_b$. Since both $W$ and $x$ have total mass $k$,
\[
\theta\coloneqq \sum_{a\in W}\lambda_a=\sum_{b\notin W}\mu_b\le k.
\]
For $a\in W$ and $b\notin W$, let $\Delta_{a,b}\coloneqq \PAV(W \setminus \{a\} \cup \{b\})-\PAV(W)$. If $\tau>0$, $\tau$-locality gives
\[
\sum_{a\in W}\sum_{b\notin W}\lambda_a\mu_b\Delta_{a,b}
<
\tau\sum_{a\in W}\sum_{b\notin W}\lambda_a\mu_b
=
\tau\theta^2
\le
\tau k^2
\le
\frac12.
\]
If $\tau=0$, exact locality gives the even stronger upper bound
\[
\sum_{a\in W}\sum_{b\notin W}\lambda_a\mu_b\Delta_{a,b}\le 0.
\]

We will now show that the same weighted sum is at least $1/2$, giving a contradiction. First, consider the case in which $u_i(W)=0$ for every voter $i$. Since $x$ strictly improves some voter, there must exist a voter who approves a candidate $b\notin W$ with $x_b>0$. Replacing any candidate of $W$ by $b$ increases the PAV score by at least $1$, contradicting $\tau\le 1/(2k^2)\le 1/2$. Hence, some voter must have positive utility under $W$.

Fix a voter $i$, and define $r_i\coloneqq u_i(W)$, $p_i\coloneqq x(A_i\cap W)$, and $q_i\coloneqq x(A_i\setminus W)$. If $r_i=0$, then the voter never loses from removing a candidate in $W$, and its contribution to the weighted swap sum is $\Gamma_i=q_i\theta\ge 0$.

Now suppose $r_i\ge 1$. Let $z_i\coloneqq r_i-p_i$, which is the total $\lambda$-weight of approved candidates in $W$. The voter gains PAV value $1/(r_i+1)$ exactly when the added candidate is approved and the removed candidate is not approved; the total $\lambda$-weight of such pairs is $q_i(\theta-z_i)$. Similarly, the voter loses PAV value $1/r_i$ exactly when the removed candidate is approved and the added candidate is not approved; the total $\lambda$-weight of such pairs is $z_i(\theta-q_i)$. Therefore, the voter's contribution is
\[ \Gamma_i = \frac{q_i(\theta-z_i)}{r_i+1} - \frac{z_i(\theta-q_i)}{r_i}.\]
Since $x$ $2$-dominates $W$, we have $p_i+q_i\ge 2r_i$, or equivalently $q_i\ge r_i+z_i$. Multiplying by $r_i(r_i+1)$ gives
\[
r_i(r_i+1)\Gamma_i
=
\theta\bigl(r_iq_i-z_i(r_i+1)\bigr)+z_iq_i.
\]
Moreover,
\[
r_iq_i-z_i(r_i+1)
\ge
r_i(r_i+z_i)-z_i(r_i+1)
=
r_i^2-z_i
\ge 0.
\]
Since $\theta\ge q_i$, it follows that
\[
r_i(r_i+1)\Gamma_i
\ge
q_i\bigl(r_iq_i-z_i(r_i+1)\bigr)+z_iq_i
=
r_iq_i(q_i-z_i).
\]
Using $q_i\ge r_i+z_i$, we obtain
\[
\Gamma_i
\ge
\frac{q_i(q_i-z_i)}{r_i+1}
\ge
\frac{(r_i+z_i)r_i}{r_i+1}
\ge
\frac{r_i^2}{r_i+1}
\ge
\frac12,
\]
where the second-last inequality uses $z_i \geq 0$ and the last inequality holds for all $r_i \geq 1$.
Thus, every voter with $r_i\ge 1$ contributes at least $1/2$, while every voter with $r_i=0$ contributes nonnegatively. Since at least one voter has $r_i\ge 1$, the total weighted swap sum is at least $1/2$, contradicting the upper bound above. Therefore, $W$ is $2$-\fPO{}.

\par\medskip
\paragraph{$2/3$-PAV approximation guarantee.}

As in the proof of weak Pareto optimality, consider a $\tau$-local-PAV committee $W$ and any other committee $W'$. Let $X \coloneqq W'\setminus W$, $Y \coloneqq W\setminus W'$, and $r \coloneqq |X|=|Y|$. Since $W' \neq W$, we have $r \geq 1$. For any $i \in N$, let $x_i = |X \cap A_i|, y_i = |Y \cap A_i|$ and $u_i = |W \cap A_i|$. For each $x \in X, y \in Y$, swapping $y$ with $x$ increases the $\PAV{}$ score by $\Delta_{y, x} < \tau$. Considering the voter-wise contributions to $\sum_{y \in Y}\sum_{x \in X} \Delta_{y, x}$, we get:
\begin{equation}
\label{eq:sum_Gamma_bound}
    \sum_{i \in N} \Gamma_i \leq \tau r^2,
\end{equation}
where $\Gamma_i = \frac{x_i(r - y_i)}{u_i + 1} - \frac{y_i (r - x_i)}{u_i}$ for all $i \in N$ and $\frac{y_i}{u_i}$ is interpreted as $0$ when $u_i$ equals $0$.
We now need the following elementary lemma about harmonic numbers, which we prove in \Cref{appendix:Proof_local_PAV_properties}.
\begin{restatable}[]{lemma}{HarmonicHelperLemma}
	For any nonnegative integers $x, y, r, u$ such that $r \ge 1$, $x, y \le r$ and $y \le u$, we have:
	\[H(u + x - y) \le \frac{3r-1}{2r}H(u) + \frac{\Gamma}{r},\]
	where $\Gamma = \frac{x(r-y)}{u + 1} - \frac{y(r - x)}{u}$ if $u > 0$ and $r x$ if $u = 0$.
	\label{lem:harmonic_numbers_lemma}
\end{restatable}
From \Cref{lem:harmonic_numbers_lemma}, we obtain that, for each voter $i \in N$,
\[H(u_i(W')) = H(u_i + x_i - y_i) \leq \frac{3r-1}{2r} H(u_i) + \frac{\Gamma_i}{r}.\]
Summing this inequality over all $i \in N$, we get:
\[
    \PAV(W') \leq \frac{3r-1}{2r}\PAV(W) + \sum_{i \in N} \frac{\Gamma_i}{r} \leq \frac{3r-1}{2r} \PAV(W) + \tau \cdot r,
\]
where the second inequality follows from~\eqref{eq:sum_Gamma_bound}. Set $\tau = \frac{1}{2k^2}$. We have $\PAV(W) \geq 1$: otherwise, no voter is represented, and adding an approved candidate in place of any member of $W$ would increase the PAV score by at least $1$, contradicting $\tau$-local optimality. Using $r \leq k$, we obtain $\tau r = \frac{r}{2k^2} \leq \frac{1}{2r} \leq \frac{\PAV(W)}{2r}$. Substituting this bound above gives $\PAV(W') \leq \frac{3}{2} \PAV(W)$ for any \localPAV{$\frac{1}{2k^2}$} committee $W$ and any arbitrary committee $W'$ of size $k$. Thus, $W$ provides a $\frac{2}{3}$-approximation to the optimal PAV score.

\begin{remark}[Tight bounds for $\tau$]
For weak Pareto optimality and \EJRplus{}, \Cref{tab:Local-PAV-Results} presents tight bounds on the parameter $\tau$ that differentiate between the existence and non-existence regimes. However, for other properties, specifically $(1+\epsilon)$-\PO{}, $2$-\fPO{}, and the $2/3$ PAV-score approximation, there is a gap in the thresholds where existence and non-existence can be established. Determining tight thresholds on $\tau$ for these properties remains an open problem.
\end{remark}

\begin{remark}[Parameter dependence]
The bounds presented in \Cref{thm:local_PAV_properties} hold for any fixed $n$ and $k$. Bounds that hold in the worst case over all $n$ and all $k$ would be overly pessimistic. In particular, a one-voter construction shows that for any fixed $\tau>0$, if $k$ is allowed to grow, no $k$-independent approximation to \fPO{} and no positive $k$-independent  approximation to the PAV score can be guaranteed (see \Cref{prop:One_Voter_Example_Instance-Dependent_Justification} in \Cref{appendix:Justifying_Instance-Dependent_Bounds}).
\label{rem:Instance-Dependence-Local-PAV}
\end{remark}

\begin{remark}[Connection to the core]
	The stability/fairness property of the \emph{core} \citep{ABC+17justified} implies weak \PO{} as a special case. \citet[Remark 4.2]{Peters25core} showed that $\tau$-local PAV satisfies the core for $k \le 7$ whenever $\tau \le 0.1 \frac{n}{k^2}$, but that it can fail the core when $k \ge 8$. The fact that local PAV tends to satisfy core-like properties might be an explanation for why it satisfies weak \PO{}.
\end{remark}

\subsection{Results for Global PAV}
\label{subsec:Results_global_PAV}

We will now analyze the efficiency properties of the global PAV (or simply PAV) rule. Every PAV committee is Pareto optimal, and is therefore also weakly Pareto optimal and $(1+\eps)$-Pareto optimal for all $\eps > 0$.  However, PAV does not satisfy \fPO{} \citep[Theorem 9]{BBF+26fractional}.

Our main result (\Cref{thm:global_PAV_properties}) concerning the PAV rule shows that every PAV committee satisfies $\alpha^\star$-\fPO{}, where $\alpha^\star \approx 1.346$ is the unique solution to $h(z) = 1$ where $h(z) = \int_{0}^z \frac{1 - e^{-y}}{y} \, dy$. We also demonstrate that this bound is tight by constructing a family of instances in which, for any $\alpha < \alpha^\star$, every PAV committee fails to satisfy $\alpha$-\fPO{}. Therefore, while PAV committees guarantee exact Pareto optimality, they only guarantee approximate fractional Pareto optimality.

\begin{restatable}[PAV is $\alpha^\star$-\fPO{}]{theorem}{GlobalPAVProperties}
Every global PAV committee is $\alpha^\star$-approximately fractionally Pareto optimal, where $\alpha^\star \approx 1.346$ is the unique solution of $\int_0^{\alpha^\star} \frac{1-e^{-y}}{y} \, dy = 1$. Moreover, this factor is asymptotically tight: For every $\alpha < \alpha^\star$, there is a family of instances with global PAV committees that are fractionally $\alpha$-Pareto dominated.
\label{thm:global_PAV_properties}
\end{restatable}

\paragraph{Proof of $\alpha^\star$-\fPO{} Guarantee for PAV.}
\label{appendix:PAV_fPO_UpperBound}

We will first show that $\Psi(x) \coloneqq \sum_{i\in N} h(u_i(x))$ lower-bounds the multilinear extension of PAV $F_{\PAV}(\cdot)$. This observation will be used to establish the upper bound in \Cref{thm:global_PAV_properties}. Note that for the global-PAV guarantee, $h$ is needed only through this pointwise lower bound; maximizing $\Psi$ becomes essential only for the polynomial-time algorithm in \Cref{subsec:Results_Our_Algorithm}.

\begin{restatable}[]{lemma}{PsiLowerBoundsMultilinearExtension}
For every fractional committee $x\in[0,1]^C$, 
$$F_{\PAV}(x) \geq \Psi(x).$$
Moreover, if $\Psi(x) \geq 1/4$, then
$$F_{\PAV}(x)-\Psi(x)\geq \frac{1}{512n^2mk^2}.$$
\label{lem:Psi_Lower_Bounds_Multilinear_Extension}
\end{restatable}

The first inequality in \Cref{lem:Psi_Lower_Bounds_Multilinear_Extension} is the PAV specialization of the Bernoulli--Poisson inequality of \citet[Lemma~2.3]{BFF21tight}. The second inequality, to our knowledge, does not immediately follow from their result. The inverse polynomial slack term in this inequality will be useful in \Cref{subsec:Results_Our_Algorithm}, as it will allow us to work with an approximately optimal solution to a convex program instead of solving it exactly, thus avoiding precision issues related to real-valued numbers.

\begin{proof}[Proof of \Cref{lem:Psi_Lower_Bounds_Multilinear_Extension}]
Fix a voter $i$ and let $u_i \coloneqq \sum_{c \in A_i} x_c$. Let $R(x)$ be the random set of candidates obtained by selecting each candidate $c$ independently with probability $x_c$. For each candidate $c \in A_i$, let $B_c$ be the
indicator random variable for the event that $c\in R(x)$. Define $S_i \coloneqq \sum_{c\in A_i} B_c$. Observe that $S_i$ is a Poisson-Binomial random variable with $\E[S_i] = u_i$.

Let $Z_i \sim \Poi(u_i)$ be a Poisson random variable with the same mean $u_i$. For every $t \in [0,1]$, it follows from independence and the inequality
$1-y \leq e^{-y}$ that
\begin{align*}
    \E[t^{S_i}]
    &= \prod_{c\in A_i}\E[t^{B_c}]\\
    &= \prod_{c\in A_i}\bigl(1-x_c(1-t)\bigr)\\
    &\leq
    \exp\left(-(1-t)\sum_{c\in A_i}x_c\right)\\
    &= e^{-(1-t)u_i}
     = \E[t^{Z_i}].
\end{align*}
From the proof of \Cref{lem:Poisson-smoothening}, we have the following integral representation of harmonic numbers: 
$$\E[H_{S_i}] = \int_0^1 \frac{1-\E[t^{S_i}]}{1-t}\,dt.$$
Thus,
\[
    \E[H_{S_i}] \geq \int_0^1 \frac{1-\E[t^{Z_i}]}{1-t}\,dt = \E[H_{Z_i}]
    = h(u_i),
\]
where the final equality follows from \Cref{lem:Poisson-smoothening}. 
After the substitution $s = 1-t$, we get
$$\E[H_{S_i}]-h(u_i) = \int_0^1 \frac{e^{-su_i}-\prod_{c\in A_i}(1-sx_c)}{s}\,ds.$$

Let $Q_i \coloneqq \sum_{c \in A_i} x_c^2$. For $0 \leq s \leq 1/k$, we have $sx_c \leq 1/2$ because $k \geq 2$. The inequality $\log(1-z) \leq -z-z^2/2$ for $0 \leq z<1$ implies
\[
 \prod_{c\in A_i}(1-sx_c)
 \leq
 \exp\!\left(-su_i-\frac{s^2Q_i}{2}\right).
\]
Since $u_i\leq k$ and $Q_i\leq u_i\leq k$, we have $e^{-su_i}\geq e^{-1}$ and $s^2Q_i/2\leq1$. Using $1-e^{-a}\geq a/2$ for $0\leq a\leq1$, we obtain
\[
\E[H_{S_i}]-h(u_i)
\geq
\int_0^{1/k}\frac{e^{-1}s^2Q_i/4}{s}\,ds
=
\frac{Q_i}{8ek^2}.
\]
Let $C_+\coloneqq\{c\in C:\supp(c)\neq\emptyset\}$ and $q\coloneqq\sum_{c\in C_+}x_c$. Summing the preceding inequality and applying Cauchy--Schwarz yields
\[
F_{\PAV}(x)-\Psi(x)
\geq
\frac{1}{8ek^2}\sum_{c\in C_+}|\supp(c)|x_c^2
\geq
\frac{q^2}{8emk^2}.
\]
Finally, $h(z)\leq z$ for $z\geq0$, and therefore
\[
\frac14\leq\Psi(x)
\leq\sum_{i\in N}u_i(x)
=\sum_{c\in C_+}|\supp(c)|x_c
\leq nq.
\]
Hence $q\geq1/(4n)$, and
\[
F_{\PAV}(x)-\Psi(x)
\geq\frac{1}{128e\,n^2mk^2}
>\frac{1}{512n^2mk^2}.\qedhere
\]
\end{proof}

We will now prove the upper bound in \Cref{thm:global_PAV_properties}, specifically that any PAV committee is $\alpha^\star$-\fPO{}.

\begin{proof} (Approximation guarantee in \Cref{thm:global_PAV_properties})
Let $W$ be a PAV committee. Suppose, for contradiction, that a feasible fractional committee $x$ $\alpha^\star$-Pareto dominates $W$. Let
$r_i\coloneqq u_i(W)$. Then $u_i(x)\ge \alpha^\star r_i$ for every voter $i$, with strict inequality for at least one voter. If $r_i\ge 1$, then
\[
h(u_i(x))\ge h(\alpha^\star r_i)\ge H_{r_i},
\]
where the first inequality follows from monotonicity of $h(\cdot)$ and the second inequality follows from \Cref{lem:Multiplicative-Shift-ineq}. If $r_i=0$, then $h(u_i(x))\ge 0=H_{r_i}$. Moreover, the corresponding inequality is strict for the voter whose fractional utility is strictly improved. Consequently,
\[
\sum_{i\in N} h(u_i(x))>\sum_{i\in N}H_{r_i}=\PAV(W).
\]

To demonstrate the contradiction, we will show that $\PAV(W) \geq \sum_{i\in N} h(u_i(x))$. 
From \Cref{prop:Pipage_for_PAV}, we know that pipage rounding any fractional committee $x$ gives a distribution over integral committees $T$ of size $k$ such that
\[
\mathbb{E}[\PAV(T)]\ge F_{\PAV}(x),
\]
where $F_{\PAV}$ is the multilinear extension of $\PAV$. Since $W$ maximizes the PAV score over all integral committees of size $k$, we have that
\[
\PAV(W)\ge \mathbb{E}[\PAV(T)].
\]
From \Cref{lem:Psi_Lower_Bounds_Multilinear_Extension}, we know that for any fractional committee $x$, \[F_{\PAV}(x) \geq \Psi(x) = \sum_{i\in N} h(u_i(x)).\]
Therefore,
\[
\PAV(W) \geq \sum_{i\in N} h(u_i(x)).
\]
This contradicts our earlier observation that $\PAV(W) < \sum_i h(u_i(x))$. Thus, no such fractional committee exists, and $W$ is $\alpha^\star$-approximately fractionally Pareto optimal.
\end{proof}

The tightness construction (see \Cref{appendix:PAV_fPO_LowerBound}) shows that the above argument is essentially the best possible. Fix any $\alpha<\alpha^\star$. We construct an instance with $k$ \emph{baseline} candidates and a large \emph{auxiliary} set. Each voter approves one baseline candidate together with a prescribed subset of auxiliary candidates. The baseline committee $W$ gives every voter utility $1$. On the other hand, a fractional committee that spreads its entire budget uniformly over the auxiliary set gives every voter utility nearly $\alpha$. Thus, $W$ is fractionally dominated by a factor arbitrarily close to $\alpha$. The key idea is to choose the auxiliary set in a complete and symmetric way so that, despite this fractional domination, $W$ remains a global PAV optimum. Indeed, if an integral committee replaces a $\lambda$ fraction of baseline candidates by auxiliary candidates, then its normalized PAV score (i.e., as a fraction of the maximum possible PAV score) converges to $h(\alpha\lambda) + (1-\lambda)\frac{1-e^{-\alpha\lambda}}{\alpha\lambda}$,  which is strictly below $1$ for every $\lambda>0$ when $\alpha<\alpha^\star$. Hence, no integral committee beats the baseline committee, even though a fractional auxiliary committee improves every voter by a factor approaching $\alpha$. Since $\alpha$ can be chosen arbitrarily close to $\alpha^\star$, the $\alpha^\star$-\fPO{} guarantee for global PAV is asymptotically tight.

\subsection{Results for the DMMS Algorithm}
\label{subsec:Results_DMMS_Algorithm}

\citetalias{DMM+21tight} prove that their algorithm (which we call the DMMS algorithm) runs in polynomial time and provides a $\gamma^\star \approx 0.79$ approximation to the optimal PAV score.
We analyze their algorithm further in \Cref{appendix:DMMS_Fails_JR}, where we also recall its formal definition. In \Cref{lem:dmms_fPO} we show that it satisfies \fPO{} and hence \PO{} and weak \PO{}. We also give an example in \Cref{lem:dmms_not_JR} where a particular execution of the DMMS algorithm fails \JR{}.

\subsection{Our Algorithm}
\label{subsec:Results_Our_Algorithm}

We will now present our main algorithmic contribution: A polynomial-time round-and-swap-based algorithm that achieves the efficiency properties of the PAV rule (\Cref{alg:Pipage-and-swap}).

\begin{algorithm}[t]
\caption{Round and Swap Algorithm}
\DontPrintSemicolon

\KwIn{An approval-based committee election instance $\I = \langle N, C, \A, k \rangle$.}
\KwOut{A committee $W\subseteq C$ with $|W|=k$.}

\BlankLine
\tcp{Concave program}
Compute an approximately optimal solution $\widehat{x}$ to the following concave program
\begin{equation}
    \max_{x\in[0,1]^C,\ \sum_{c\in C}x_c=k} \Psi(x)\coloneqq \sum_{i\in N} h(u_i(x)),
\label{eq:smooth-program}
\end{equation}
where $\Psi(\widehat{x}) \geq \Psi^\star - \eta$, $\Psi^\star \coloneqq \max_{x}\Psi(x)$ and $\eta \coloneqq \frac{1}{1024n^2mk^2}$.

Let $F_{\PAV}$ denote the multilinear extension of the $\PAV(\cdot)$ set function.\;

\BlankLine
\tcp{Rounding step}
$W\gets \textsc{PipageRound}_{F_{\PAV}}(\widehat{x})$\tcp*{$W$ is an integral committee of size $k$.}

\BlankLine
\tcp{Local search phase}
\While{swapping $c \in W$ and $c' \notin W$ improves PAV score by at least $\frac{1}{2k^2}$}{
    $W \leftarrow W \setminus \{c\} \cup \{c'\}$
}

\Return{$W$}\;
\label{alg:Pipage-and-swap}
\end{algorithm}

\begin{restatable}[\EJRplus{}, $1.346$-\fPO{}, and $0.79$-$\PAV$ in poly time]{theorem}{MainAlgo}
Given any instance of approval-based committee voting, \Cref{alg:Pipage-and-swap} runs in polynomial time and returns a committee $W$ satisfying \EJRplus{} and $\alpha^\star$-\fPO{}, where $\alpha^\star \approx 1.346$ is the unique solution of $h(\alpha^\star) = \int_0^{\alpha^\star}\frac{1-e^{-y}}{y}\,dy=1$. The same committee is weakly Pareto optimal, is $(1+\eps)$-\PO{} for every $\eps>0$, and achieves a $\gamma^\star \approx 0.79$ approximation to the optimal PAV score, where
$\gamma^\star \coloneqq h(1) = \int_0^1\frac{1-e^{-y}}{y}\,dy$.
\label{thm:Main-Algo}
\end{restatable}

Our algorithm (\Cref{alg:Pipage-and-swap}) consists of three phases:
\begin{enumerate}
    \item \emph{Approximate concave optimization.} The algorithm starts by computing a fractional committee $\widehat{x}$ that approximately optimizes a smooth concave surrogate $\Psi(x)=\sum_{i\in N} h(u_i(x))$; specifically, $\Psi(\widehat{x}) \geq \Psi^\star - \eta$ where $\Psi^\star \coloneqq \max_x \Psi(x)$. The function $h$ is chosen to interact well with both fractional Pareto optimality and the PAV objective; in particular, $h(\alpha^\star t)\geq H_t$ for every integer $t\geq 0$, where $h(\alpha^\star)=1$. It can be shown that the multilinear extension $F_{\PAV}$ at $\widehat{x}$ is strictly larger than the corresponding surrogate, i.e., $F_{\PAV}(\widehat{x}) \geq \Psi(\widehat{x}) + 2\eta$. Because $\eta$ has polynomial bit complexity and the function $h$ and its derivative $h'$ can be evaluated to polynomial accuracy, standard convex-optimization methods compute the required $\widehat{x}$ in polynomial time.
    \item \emph{Pipage rounding.} After computing an approximately optimal solution $\widehat{x}$ to the convex program (to within an inverse polynomial additive accuracy $\eta$), the algorithm uses pipage rounding with respect to $F_{\PAV}$ to compute an integral committee $W$ of size $k$. The PAV score of this committee is guaranteed to be at least $F_{\PAV}(\widehat{x})$, which, by the aforementioned property of $\Psi$, is at least $\Psi^\star$. 
    \item \emph{Local search.} The third phase of the algorithm involves a local search strategy in which it performs PAV score-improving candidate swaps. Specifically, starting with the integral committee $W$ returned by the rounding phase, the algorithm performs swaps between elected and unelected candidates as long as the PAV score improves by at least $\frac{1}{2k^2}$. Since the maximum possible PAV score is at most $nH_k$, the number of swaps is polynomial, and the process terminates with a \localPAV{$\frac{1}{2k^2}$} committee.
\end{enumerate}

Due to \Cref{thm:local_PAV_properties}, the committee returned by the algorithm inherits the positive guarantees of local PAV outcomes; in particular, it satisfies \EJRplus{}, weak \PO{}, and $(1+\eps)$-\PO{} for all $\eps > 0$. The same theorem also readily gives the $2$-\fPO{} and $2/3$ $\PAV{}$-score approximation guarantees. However, due to our careful initialization of the local search phase, we can obtain strictly better approximation factors, as we discuss next.

Note that the monotonicity of PAV score during the local search phase preserves the PAV approximation obtained after pipage rounding. Specifically, for an optimal PAV committee $O$, the integral point $\mathbf{1}_O$ is feasible for the smooth program and the inequality $h(t)\geq \gamma^\star H_t$ for all integers $t\geq 0$, with $\gamma^\star=h(1)$, gives $\Psi^\star \geq \Psi(\mathbf{1}_O) \geq \gamma^\star \PAV(O)$. Therefore, the final committee, whose PAV score is at least $\Psi^\star$, achieves a $\gamma^\star$-approximation to the optimal PAV score.

Using similar reasoning, we can show that the algorithm's outcome $W$ satisfies $\alpha^\star$-\fPO{}. Indeed, if $W$ were fractionally $\alpha^\star$-Pareto dominated by some fractional committee $y$, then the inequality $h(\alpha^\star t)\geq H_t$ would imply $\Psi(y)>\PAV(W)$. However, this contradicts our earlier observation that $\PAV(W) \geq \Psi^\star$,
yielding the desired $\alpha^\star$-\fPO{} guarantee.

Let us now formalize the above arguments. We will prove \Cref{thm:Main-Algo} in three parts: \Cref{lem:swap-terminates} will establish \EJRplus{}, \Cref{lem:alpha-fpo} will show the $\alpha^\star$-\fPO{} guarantee, and \Cref{lem:pav-approx} will prove the $\gamma^\star$-approximation of PAV score.

\begin{lemma}[\EJRplus{} guarantee]
The local search phase of \Cref{alg:Pipage-and-swap} terminates after at most $2nk^2 H_k$ swaps. The final committee $W$ satisfies \EJRplus{} and $\PAV(W)\ge \Psi(x^\star)$, where $x^\star$ is the fractional committee maximizing $\Psi(x)$.
\label{lem:swap-terminates}
\end{lemma}

\begin{proof}[Proof of \Cref{lem:swap-terminates}]
Every swap step increases PAV by at least $1/2k^2$. Also $\PAV(W)\le nH_k$ for every committee $W$. Thus, there can be at most $2nk^2H_k$ swap steps. At termination, the resulting committee $W$ is \localPAV{$\frac{1}{2k^2}$}, and therefore satisfies \EJRplus{}~(\Cref{thm:local_PAV_properties}).

Let $W_0$ denote the committee returned by the (deterministic) pipage rounding step. Then, by \Cref{prop:Pipage_for_PAV}, we know that $\PAV(W_0) \geq F_{\PAV}(\widehat{x})$. Additionally, from \Cref{lem:Psi_Lower_Bounds_Multilinear_Extension}, we know that if $\Psi(x) \geq 1/4$, then $F_{\PAV}(x) \geq \Psi(x) + 2\eta$, with $\eta = \frac{1}{1024n^2mk^2}$. The assumption $\Psi(\widehat{x}) \geq 1/4$ holds without loss of generality since, for any instance where at least one voter approves at least one candidate, the optimum objective value is $\Psi^\star = \Psi(x^\star) > 3/4$, and by the choice of $\eta$, $\Psi(\widehat{x}) \geq 1/4$. Therefore, we have $\PAV(W_0) \geq F_{\PAV}(\widehat{x}) \geq \Psi(\widehat{x}) + 2\eta \geq \Psi^\star + \eta > \Psi^\star$. Finally, since the PAV score only increases during the local search phase, we get that $\PAV(W)\ge \PAV(W_0)\ge \Psi(x^\star)$.
\end{proof}

We will now show that the output of \Cref{alg:Pipage-and-swap} satisfies the $\alpha^\star$-\fPO{} guarantee. The argument is similar to that in the proof of~\Cref{thm:global_PAV_properties}.

\begin{lemma}[$\alpha^\star$-\fPO{} guarantee]
The committee $W$ output by \Cref{alg:Pipage-and-swap} is $\alpha^\star$-\fPO{}.
\label{lem:alpha-fpo}
\end{lemma}

\begin{proof}[Proof of \Cref{lem:alpha-fpo}]
Suppose, for contradiction, that some fractional committee $y$ $\alpha^\star$-dominates $W$. Thus $u_i(y)\ge \alpha^\star u_i(W)$ for all voters $i$ and the inequality is strict for at least one voter. By monotonicity of $h$ and \Cref{lem:Multiplicative-Shift-ineq}, $h(u_i(y)) \geq h(\alpha^\star u_i(W)) \geq H_{u_i(W)}$ for all voters $i$. Moreover, the sum is strict. If the strict domination occurs for a voter with $u_i(W)>0$, then the strict monotonicity of $h$ gives strictness. If it occurs for a voter with $u_i(W)=0$, then $u_i(y)>0$, so $h(u_i(y))>0=H_0$. Therefore
\[
  \Psi(y)=\sum_i h(u_i(y))>
  \sum_i H_{u_i(W)}=\PAV(W).
\]
From \Cref{lem:swap-terminates}, $\PAV(W)\ge \Psi(x^\star)$.
Combining the two inequalities gives
\[
  \Psi(y) > \PAV(W) \geq \Psi(x^\star),
\]
which contradicts the optimality of $x^\star$. Thus, no fractional committee $\alpha^\star$-dominates $W$.
\end{proof}

\begin{lemma}[$\gamma^\star$-PAV guarantee]
Let $W^{\mathrm{opt}}$ be any PAV committee. The committee $W$ returned by \Cref{alg:Pipage-and-swap} satisfies \[\PAV(W)\ge \gamma^\star\,\PAV(W^{\mathrm{opt}}),\] where $\gamma^\star \approx 0.79$ is such that $\gamma^\star = h(1)$.
\label{lem:pav-approx}
\end{lemma}

\begin{proof}[Proof of \Cref{lem:pav-approx}]
By definition of $\Psi(x^\star)$, we have that $\Psi(x^\star) \geq \Psi(\mathbf{1}_{W^{\mathrm{opt}}})$, where $\mathbf{1}_{W^{\mathrm{opt}}}$ denotes the characteristic vector corresponding to the integral committee $W^{\mathrm{opt}}$. By \Cref{lem:h_Lower_and_Upper_Bound_Harmonics},
\[
  \Psi(\mathbf{1}_{W^{\mathrm{opt}}})
  =\sum_i h(u_i(W^{\mathrm{opt}}))
  \ge \gamma^\star\sum_i H_{u_i(W^{\mathrm{opt}})}
  =\gamma^\star\,\PAV(W^{\mathrm{opt}}).
\]
Finally, by \Cref{lem:swap-terminates}, $\PAV(W) \geq \Psi(x^\star)$. Combining the above inequalities yields the claimed approximation guarantee.
\end{proof}

Compared to the DMMS algorithm, our algorithm rounds a different fractional committee, and it adds a local search phase. It is natural to ask whether keeping the DMMS fractional committee and adding the local search phase would also work. Certainly the resulting committee will satisfy EJR+. However, in \Cref{lem:dmms_not_JR}, we give an example where there is a run in which the resulting committee may fail 1.346-\fPO{}, and indeed satisfy no better than 2-\fPO{}. Since the DMMS algorithm without local search satisfies \fPO{}, this also shows that adding a local search phase can ``damage'' \fPO{} by a factor of at least 2.

\section{Conclusion and Future Directions}

We studied the interplay between representation and efficiency in approval-based committee voting, with a particular focus on local and global variants of PAV. We used multiplicative approximate fractional Pareto optimality  ($\alpha$-fPO) as a quantitative efficiency measure, established the tight $\alpha^\star$-fPO threshold for global PAV, and showed that local PAV has weaker but still nontrivial global efficiency guarantees.
Our round-and-swap algorithm combines approximate concave optimization, pipage rounding, and PAV-improving local search. The resulting committee simultaneously satisfies \EJRplus{}, weak \PO{}, $(1+\eps)$-\PO{} for every $\eps > 0$, the $\alpha^\star$-fPO guarantee of global PAV, and the optimal polynomial-time approximation to PAV score achievable unless $\textrm{\textup{P}}=\NP{}$.

In light of the recent impossibility result of \citet{jin2026limitationsbobw} concerning the compatibility of \JR{} and \fPO{}, a natural direction is to determine the sharp threshold $\alpha$ for approximate fractional Pareto optimality: namely, the value of $\alpha > 1$ such that every approval-based multiwinner voting instance admits a committee satisfying both \JR{} and $\alpha$-\fPO{}, whereas, for every fixed $\eps>0$, some instance admits no committee satisfying both \JR{} and $(\alpha-\eps)$-\fPO{}. Our algorithm gives an upper bound of $\alpha^\star \approx 1.346$ even for \EJRplus{}, where the recent impossibility rules out $\alpha = 1$ even with the weaker requirement of \JR{}. The computational aspects of this question also remain largely unexplored. Whenever existence is guaranteed, can such a committee be computed efficiently? More broadly, do stronger representation and efficiency guarantees entail inherent computational barriers?

A complementary direction is to revisit the failed approaches in (\Cref{appendix:Limitations}). Although none can guarantee exact \JR{}+fPO in general, some may yield \JR{} or \EJRplus{} together with an $\alpha$-fPO factor on general or restricted domains, potentially improving on $\alpha^\star$ when a weaker representation axiom is used or when additional structure is present.

\section*{Usage of Artificial Intelligence}
Many of the proofs and examples in this paper were obtained by or with the help of LLMs. 
We see our own technical contributions in the questions we asked.
For example, we asked whether local PAV satisfies weak PO, a question that was prompted by our difficulty in finding counterexamples; GPT-5 provided the (in retrospect simple) proof in October 2025.
Another conceptual contribution is the multiplicative relaxation of \fPO{} and asking about the approximation factor that PAV achieves. We also obtained the method to achieve the same guarantee in polynomial time.
All of the results and techniques in \Cref{appendix:Limitations} were also obtained and designed by us (but not all of the counterexamples).
The other results were found with help from GPT-5.5 and GPT-5.6 Sol.
The writing and presentation of the results was almost entirely done by us.

\section*{Acknowledgments}
We are grateful to Jannik Peters for helpful discussions and pointers to the literature. RV thanks Harshal Singh Sindal, Aditya Prakash, and Pravar Kataria for their input during early stages of this work. Additionally, RV acknowledges support from ANRF grant no. CRG/2022/002621, DST
INSPIRE grant no.\ DST/INSPIRE/04/2020/000107, and Mr.\ D.P. Gupta Chair Professorship.
DP was funded in part by the Agence Nationale de la Recherche as part of the France 2030 program under grant ANR-23-IACL-0008 (PR[AI]RIE-PSAI).
Parts of this work benefited from work done at the Dagstuhl seminars ``Fair Division: Algorithms, Solution Concepts, and Applications'' (24401) and ``Randomized Rounding in Algorithms, Statistics, and Economics and Computation'' (26242).

\bibliographystyle{ACM-Reference-Format}
\providecommand\showeprint[2][arxiv]{%
	arXiv:\href{https://arxiv.org/abs/#2}{#2}%
}
\bibliography{References}

@misc{jin2026limitationsbobw,
    title={Limitations of Best-of-Both-Worlds Solutions in Approval-Based Multiwinner Elections}, 
    author={Jiarong Jin and Xuanxuan Liu and Biaoshuai Tao},
    year={2026},
    eprint={2608.01830},
    archivePrefix={arXiv},
    primaryClass={cs.GT},
    url={https://arxiv.org/abs/2608.01830}, 
}

@inproceedings{aziz2023bestofbothworlds,
  author = {Haris Aziz and Xinhang Lu and Mashbat Suzuki and Jeremy Vollen and Toby Walsh},
  title = {{Best-of-Both-Worlds Fairness in Committee Voting}},
  booktitle = {Proceedings of the 19th International Conference on Web and Internet Economics (WINE)},
  pages = {676},
  year = {2023},
  primaryclass = {cs.GT},
  eprint = {2303.03642},
  archiveprefix = {arXiv},
}

@article{AFSV23bobw,
  author = {Aziz, Haris and Freeman, Rupert and Shah, Nisarg and Vaish, Rohit},
  title = {Best of Both Worlds: Ex Ante and Ex Post Fairness in Resource Allocation},
  journal = {Operations Research},
  volume = {72},
  number = {4},
  pages = {1674--1688},
  year = {2023},
  doi = {10.1287/opre.2022.2432},
}

@article{CG18nashsocialwelfare,
  author = {Cole, Richard and Gkatzelis, Vasilis},
  title = {Approximating the {Nash} Social Welfare with Indivisible Items},
  journal = {SIAM Journal on Computing},
  volume = {47},
  number = {3},
  pages = {1211--1236},
  year = {2018},
  doi = {10.1137/15M1053682},
}

@inproceedings{KP25lindahl,
  author = {Christian Kroer and Dominik Peters},
  title = {Computing {L}indahl Equilibrium for Public Goods with and without Funding Caps},
  booktitle = {Proceedings of the 26th ACM Conference on Economics and Computation (EC)},
  pages = {129},
  year = {2025},
  doi = {10.1145/3736252.3742510},
  archiveprefix = {arXiv},
  eprint = {2503.16414},
}

@inproceedings{Peters25core,
  author = {Dominik Peters},
  title = {{The Core of Approval-Based Committee Elections with Few Seats}},
  booktitle = {Proceedings of the 34th International Joint Conference on Artificial Intelligence (IJCAI)},
  pages = {4014--4022},
  year = {2025},
  doi = {10.24963/ijcai.2025/447},
}

@article{EFF20pricepareto,
  author = {Edith Elkind and Angelo Fanelli and Michele Flammini},
  title = {{Price of Pareto Optimality in Hedonic Games}},
  journal = {Artificial Intelligence},
  volume = {288},
  pages = {103357},
  year = {2020},
  doi = {10.1016/j.artint.2020.103357}
}

@article{BGPSW24approvalapportionment,
  author = {Markus Brill and Paul G{\"{o}}lz and Dominik Peters and Ulrike Schmidt{-}Kraepelin and Kai Wilker},
  title = {{Approval-Based Apportionment}},
  journal = {Mathematical Programming},
  volume = {203},
  number = {1},
  pages = {77--105},
  year = {2024},
  doi = {10.1007/s10107-022-01852-1},
}

@article{YW23parameterized,
  author = {Yang, Yongjie and Wang, Jianxin},
  title = {{Parameterized Complexity of Multiwinner Determination: More Effort Towards Fixed-Parameter Tractability}},
  journal = {Autonomous Agents and Multi-Agent Systems},
  volume = {37},
  number = {2},
  pages = {28},
  year = {2023},
  publisher = {Springer},
  doi = {10.1007/s10458-023-09610-z},
}

@inproceedings{AGG+15computationalaspects,
  author = {Aziz, Haris and Gaspers, Serge and Gudmundsson, Joachim and Mackenzie, Simon and Mattei, Nicholas and Walsh, Toby},
  title = {Computational Aspects of Multi-Winner Approval
                   Voting},
  booktitle = {Proceedings of the 14th International Conference on
                   Autonomous Agents and Multiagent Systems (AAMAS)},
  pages = {107--115},
  year = {2015},
  publisher = {IFAAMAS},
  url = {https://www.ifaamas.org/Proceedings/aamas2015/aamas/p107.pdf},
}

@article{Thie1895,
  author = {Thorvald N. Thiele},
  title = {Om Flerfoldsvalg},
  journal = {Oversigt over det Kongelige Danske Videnskabernes
             Selskabs Forhandlinger},
  pages = {415--441},
  year = {1895},
  url_hidden = {https://dominik-peters.de/archive/thiele1895.pdf},
}

@article{GM24fPO,
  title={{Computing Pareto-Optimal and Almost Envy-Free Allocations of Indivisible Goods}},
  author={Garg, Jugal and Murhekar, Aniket},
  journal={Journal of Artificial Intelligence Research},
  volume={80},
  pages={1--25},
  year={2024},
  doi={10.1613/jair.1.15414}
}

@article{BMS05,
  author = {Anna Bogomolnaia and Herv\'e Moulin and Richard Stong},
  title = {{Collective Choice under Dichotomous Preferences}},
  journal = {Journal of Economic Theory},
  volume = {122},
  number = {2},
  pages = {165--184},
  year = {2005},
  doi = {10.1016/j.jet.2004.05.005},
}

@inproceedings{BKV18ef1,
  author = {Barman, Siddharth and Krishnamurthy, Sanath Kumar and Vaish, Rohit},
  title = {{Finding Fair and Efficient Allocations}},
  booktitle = {Proceedings of the 2018 ACM Conference on Economics and Computation (EC)},
  pages = {557--574},
  year = {2018},
  doi = {10.1145/3219166.3219176}
}

@article{ABC+17justified,
  title={{Justified Representation in Approval-Based Committee Voting}},
  author={Aziz, Haris and Brill, Markus and Conitzer, Vincent and Elkind, Edith and Freeman, Rupert and Walsh, Toby},
  journal={Social Choice and Welfare},
  year={2017},
  volume={48},
  number={2},
  pages={461--485},
  publisher={Springer},
  doi = {10.1007/s00355-016-1019-3},
}

@inproceedings{BP23robust,
author = {Brill, Markus and Peters, Jannik},
title = {{Robust and Verifiable Proportionality Axioms for Multiwinner Voting}},
year = {2023},
booktitle = {Proceedings of the 24th ACM Conference on Economics and Computation (EC)},
pages = {301},
publisher = {Association for Computing Machinery},
doi = {10.1145/3580507.3597785},
}

@book{LS23book,
    author = {Lackner, Martin and Skowron, Piotr},
    title = {{Multi-Winner Voting with Approval Preferences}},
    year={2023},
    publisher={Springer},
    doi = {10.1007/978-3-031-09016-5},
}

@inproceedings{SV24maximumflowfairnetwork,
    author = {Suzuki, Mashbat and Vollen, Jeremy},
    title = {{Maximum Flow is Fair: A Network Flow Approach to Committee Voting}},
    year = {2024},
    booktitle = {Proceedings of the 25th ACM Conference on Economics and Computation (EC)},
    pages = {964--983},
    doi = {10.1145/3670865.3673603},
}

@article{AM20computing,
  title={{Computing and Testing Pareto Optimal Committees}},
  author={Aziz, Haris and Monnot, J{\'e}r{\^o}me},
  journal={Autonomous Agents and Multi-Agent Systems},
  volume={34},
  pages={1--20},
  year={2020},
  publisher={Springer},
  doi={10.1007/s10458-020-09445-y}
}

@article{LS20utilitarian,
  title={{Utilitarian Welfare and Representation Guarantees of Approval-Based Multiwinner Rules}},
  author={Lackner, Martin and Skowron, Piotr},
  journal={Artificial Intelligence},
  volume={288},
  pages={103366},
  year={2020},
  publisher={Elsevier},
  doi = {10.1016/j.artint.2020.103366},
}

@inproceedings{BBP+21distribution,
  title={{Distribution Rules under Dichotomous Preferences: Two Out of Three Ain't Bad}},
  author={Brandl, Florian and Brandt, Felix and Peters, Dominik and Stricker, Christian},
  booktitle = {Proceedings of the 22nd ACM Conference on Economics and Computation (EC)},
  pages={158--179},
  year={2021},
  doi = {10.1145/3465456.3467653},
}

@inproceedings{PS20proportionality,
  title={{Proportionality and the Limits of Welfarism}},
  author={Peters, Dominik and Skowron, Piotr},
  booktitle={Proceedings of the 21st ACM Conference on Economics and Computation},
  pages={793--794},
  year={2020},
  doi = {10.1145/3391403.3399465},
  eprint = {1911.11747},
  archiveprefix = {arXiv},
}

@article{PL20spoc,
  author = {Peters, Dominik and Lackner, Martin},
  title = {Preferences Single-Peaked on a Circle},
  journal = {Journal of Artificial Intelligence Research},
  volume = {68},
  pages = {463--502},
  year = {2020},
  doi = {10.1613/jair.1.11732},
}

@inproceedings{DMM+21tight,
  title={{Tight Approximation for Proportional Approval Voting}},
  author={Dudycz, Szymon and Manurangsi, Pasin and Marcinkowski, Jan and Sornat, Krzysztof},
  booktitle={Proceedings of the 29th International Joint Conference on Artificial Intelligence (IJCAI)},
  pages={276--282},
  year={2020},
  doi={10.24963/ijcai.2020/39}
}

@inproceedings{KE24lower,
  title={{A Lower Bound for Local Search Proportional Approval Voting}},
  author={Kraiczy, Sonja and Elkind, Edith},
  booktitle = {Proceedings of the 32nd Annual European Symposium on Algorithms (ESA)},
  pages = {82:1--82:14},
  year={2024},
  doi = {10.4230/LIPIcs.ESA.2024.82},
  organization={Schloss Dagstuhl--Leibniz-Zentrum f{\"u}r Informatik}
}

@inproceedings{AEH+18complexity,
  title={{On the Complexity of Extended and Proportional Justified Representation}},
  author={Aziz, Haris and Elkind, Edith and Huang, Shenwei and Lackner, Martin and S{\'a}nchez-Fern{\'a}ndez, Luis and Skowron, Piotr},
  booktitle = {Proceedings of the 32nd AAAI Conference on Artificial Intelligence (AAAI)},
  pages = {902--909},
  year={2018},
  doi = {10.1609/aaai.v32i1.11478},
}

@misc{S26pareto,
  title={{Pareto Optimality in Approval-Based Multiwinner Voting}},
  author={Sch{\"u}nke, Joshua},
  eprint={2605.30490},
  archivePrefix={arXiv},
  primaryClass={cs.GT},
  year={2026}
}

@misc{BBF+26fractional,
  title={{Fractional Pareto-Optimality in Multiwinner Voting}},
  author={Becker, Patrick and Boehmer, Niclas and Frank, Fabian and Glessen, Lara},
  eprint={2606.11160},
  archivePrefix={arXiv},
  primaryClass={cs.GT},
  year={2026}
}

@misc{ALS+26computing,
  title={{Computing Thiele Rules on Interval Elections and Their Generalizations}},
  author={Avramidis, Dimitris and Lassota, Alexandra and Schmidt-Kraepelin, Ulrike and Vetta, Adrian},
  eprint={2605.03067},
  archivePrefix={arXiv},
  primaryClass={cs.AI},
  year={2026}
}

@misc{MS26polynomial,
  title={{Polynomial-Time Algorithm for Thiele Voting Rules with Voter Interval Preferences}},
  author={Manurangsi, Pasin and Sornat, Krzysztof},
  eprint={2604.05953},
  archivePrefix={arXiv},
  primaryClass={cs.GT},
  year={2026}
}

@article{AS04pipage,
  title={{Pipage Rounding: A New Method of Constructing Algorithms with Proven Performance Guarantee}},
  author={Ageev, Alexander A. and Sviridenko, Maxim I.},
  journal={Journal of Combinatorial Optimization},
  volume={8},
  number={3},
  pages={307--328},
  year={2004},
  publisher={Springer},
  doi={10.1023/B:JOCO.0000038913.96607.c2}
}

@article{CCP+11maximizing,
  title={{Maximizing a Monotone Submodular Function Subject to a Matroid Constraint}},
  author={Calinescu, Gruia and Chekuri, Chandra and Pal, Martin and Vondr{\'a}k, Jan},
  journal={SIAM Journal on Computing},
  volume={40},
  number={6},
  pages={1740--1766},
  year={2011},
  publisher={SIAM},
  doi={10.1137/080733991}
}

@article{SFL16finding,
  title={{Finding a Collective Set of Items: From Proportional Multirepresentation to Group Recommendation}},
  author={Skowron, Piotr and Faliszewski, Piotr and Lang, J{\'e}r{\^o}me},
  journal={Artificial Intelligence},
  volume={241},
  pages={191--216},
  year={2016},
  publisher={Elsevier},
  doi = {10.1016/j.artint.2016.09.003},
}

@inproceedings{BFF21tight,
  title={{Tight Approximation Guarantees for Concave Coverage Problems}},
  author={Barman, Siddharth and Fawzi, Omar and Ferm{\'e}, Paul},
  booktitle={Proceedings of the 38th International Symposium on Theoretical Aspects of Computer Science (STACS)},
  pages={9:1--9:17},
  year={2021},
  doi={10.4230/LIPIcs.STACS.2021.9}
}

@misc{DLMF, 
    author = {{DLMF}},
    title = {{NIST Digital Library of Mathematical Functions}},
    howpublished = {\url{https://dlmf.nist.gov/6}},
    note = {Release 1.2.7 of 2026-06-15. Chapter 6: Exponential, Logarithmic, Sine, and Cosine Integrals},
    editor = {Olver, F. W. J. and Olde Daalhuis, A. B. and Lozier, D. W. and Schneider, B. I. and Boisvert, R. F. and Clark, C. W. and Miller, B. R. and Saunders, B. V. and Cohl, H. S. and McClain, M. A.},
    year={2026}
}

@inproceedings{ILW+17approximate,
  title={{Approximate Efficiency in Matching Markets}},
  author={Immorlica, Nicole and Lucier, Brendan and Weyl, Glen and Mollner, Joshua},
  booktitle={Proceedings of the 13th International Conference on Web and Internet Economics (WINE)},
  pages={252--265},
  year={2017},
  organization={Springer},
  doi={10.1007/978-3-319-71924-5_18}
}

@inproceedings{PY00approximability,
  title={{On the Approximability of Trade-Offs and Optimal Access of Web Sources}},
  author={Papadimitriou, Christos H and Yannakakis, Mihalis},
  booktitle={Proceedings of the 41st Annual Symposium on Foundations of Computer Science},
  pages={86--92},
  year={2000},
  organization={IEEE},
  doi={10.1109/SFCS.2000.892068}
}

@inproceedings{M25polynomial,
  title={{A Polynomial-Time Algorithm for Fair and Efficient Allocation with a Fixed Number of Agents}},
  author={Mahara, Ryoga},
  booktitle={Proceedings of the 21st International Conference on Web and Internet Economics},
  pages={413--430},
  year={2025},
  organization={Springer},
  doi={10.1007/978-3-032-18660-7_22}
}

@article{ALM+19strategyproof,
  author  = {Aziz, Haris and Lev, Omer and Mattei, Nicholas and Rosenschein, Jeffrey S. and Walsh, Toby},
  title   = {{Strategyproof Peer Selection Using Randomization, Partitioning, and Apportionment}},
  journal = {Artificial Intelligence},
  volume  = {275},
  pages   = {295--309},
  year    = {2019},
  doi     = {10.1016/j.artint.2019.06.004}
}

\clearpage

\appendix
\addtocontents{toc}{\protect\setcounter{tocdepth}{1}}

\begin{center}
    \Large{Appendix}    
\end{center}

\section{Omitted Material from Section~\ref{sec:Preliminaries}}
\label{appendix:Omitted_Preliminaries}

\subsection{Connecting $h(\cdot)$ and Harmonic Numbers (Lemmas~\ref{lem:Poisson-smoothening},~\ref{lem:h_Lower_and_Upper_Bound_Harmonics}, and~\ref{lem:Multiplicative-Shift-ineq})}
\label{appendix:h_and_harmonics}

\PoissonSmoothening*
\begin{proof}
We know that, for every integer $s\geq 0$ and every $t \in [0,1)$,
$$\frac{1-t^s}{1-t}=1+t+\cdots+t^{s-1}.$$ 
By integrating both sides over $[0,1)$, we get that 
$$H_s = \int_0^1 \frac{1-t^s}{1-t}\,dt.$$
If $Z\sim\operatorname{Poisson}(\mu)$,
then $\E[t^Z]=e^{-\mu(1-t)}$, and hence
\[
    \E[H_Z]
    =\int_0^1 \frac{1-\E[t^Z]}{1-t}\,dt
    =\int_0^1 \frac{1-e^{-\mu(1-t)}}{1-t}\,dt
    =\int_0^\mu \frac{1-e^{-y}}{y}\,dy
    =h(\mu). \qedhere
\]
\end{proof}

\hLowerAndUpperBoundHarmonics*
\begin{proof}
We will start by showing that for every integer $t \geq 1$, $h(t) < H_t$. Let $X \sim \Poi(t)$ be a Poisson random variable with mean $t$. From \Cref{lem:Poisson-smoothening}, we know that $h(t) = \mathbb{E}[H_{X}]$.

Let $\overline H:[0,\infty)\to\mathbb{R}$ be the piecewise-linear interpolation of the harmonic numbers, i.e., $\overline H(m)=H_m$ for every integer $m\ge 0$, and $\overline H$ is linear on each interval $[m,m+1]$. The slope of $\overline H$ on $[m,m+1]$ is $1/(m+1)$, so these slopes are nonincreasing. Thus, $\overline H$ is concave.

Since the right slope of $\overline H$ at $t$ is $1/(t+1)$, concavity gives
\[
\overline H(x)\le H_t+\frac{x-t}{t+1}
\]
for all $x\ge 0$. The inequality is strict whenever $x<t$, because all slopes of $\overline H$ on intervals to the left of $t$ are strictly larger than $1/(t+1)$. Since $\mathbb{P}[X=0]=e^{-t}>0$ and $0<t$, taking expectations gives
\[
\mathbb{E}[H_X]
=\mathbb{E}[\overline H(X)]
< H_t+\frac{\mathbb{E}[X]-t}{t+1}
=H_t.
\]
Using $h(t)=\mathbb{E}[H_X]$, we conclude that $h(t)<H_t$.

We will now show that for every integer $t \geq 1$, $h(t) \geq \gamma^\star H_t$. For $t\ge 1$, write
\[
  h(t)=\sum_{r=1}^t\int_{r-1}^{r}\frac{1-e^{-y}}{y}\,dy.
\]
The first interval contributes exactly $\gamma^\star$. The function
\[
  \phi(y)\coloneqq\frac{1-e^{-y}}{y},\qquad \phi(0)\coloneqq1,
\]
is decreasing on $[0,\infty)$. Thus, for every $r\ge 1$,
\[
  \int_{r-1}^{r}\frac{1-e^{-y}}{y}\,dy
  \ge \frac{1-e^{-r}}{r}.
\]
For the summands beyond the first one, i.e., for $r\ge 2$, we have
\[
  1-e^{-r}\ge 1-e^{-2}>\gamma^\star.
\]
Hence, for every $r\ge 2$,
\[
  \int_{r-1}^{r}\frac{1-e^{-y}}{y}\,dy
  \ge \frac{\gamma^\star}{r}.
\]
Combining this with the exact first-interval contribution $\gamma^\star=\gamma^\star/1$, and summing over $r \in \{1,\dots,t\}$, gives $h(t)\ge \gamma^\star H_t$.
\end{proof}

\MultiplicativeShiftLemma*

\begin{proof}
For $t=0$, both sides are zero, and for $t=1$, the claim is exactly $h(\alpha^\star) = 1 = H_1$, which holds by definition of $\alpha^\star$. Therefore, for the remainder of the proof, we will assume that $t \geq 2$. 

First, observe that $\alpha^\star > 4/3$. This is because the alternating Taylor expansion of $e^{-y}$ gives
\[\frac{1-e^{-y}}{y} < 1 - \frac{y}{2} + \frac{y^2}{6} - \frac{y^3}{24} + \frac{y^4}{120}, \qquad y > 0.\]
Hence,
\[ h \! \left(\frac{4}{3} \right) = \int_{0}^{4/3} \frac{1-e^{-y}}{y} < \frac{4}{3} - \frac{4}{9} +\frac{32}{243} - \frac{8}{243} + \frac{128}{18225} < 1.\]
Since $h$ is strictly increasing, the solution of $h(z) = 1$ must satisfy
$z > 4/3$.

For every integer $r\ge 1$,
\begin{align}
  h(\alpha^\star(r+1)) - h(\alpha^\star r) & = \int_r^{r+1} \frac{1-e^{-\alpha^\star s}}{s} \, ds \nonumber \\
  & \geq \int_r^{r+1} \frac{1-e^{-4s/3}}{s} \, ds \nonumber \\
  & \geq (1-e^{-4r/3}) \ln \left( 1+\frac{1}{r} \right).
  \label{temp:lower-bound}
\end{align}

We will now show that for every integer $r \geq 1$,
\begin{equation}
    h(\alpha^\star(r+1)) - h(\alpha^\star r) \geq \frac{1}{r+1}.
    \label{eqn:temp-telescoping}
\end{equation}
Note that by summing the inequalities in \Cref{eqn:temp-telescoping} from $r=1$ to $t-1$, and using $h(\alpha^\star)=1$, we will obtain the desired guarantee, that is,
\[h(\alpha^\star t) \geq 1 + \sum_{r=1}^{t-1}\frac1{r+1} = H_t.\]

To prove \Cref{eqn:temp-telescoping}, first observe that the right-hand side in \Cref{temp:lower-bound} for $r=1$ is at least $1/2$. Indeed, the elementary estimates $e^{-4/3}<11/40$ and $\ln 2>69/100$ give $(1-e^{-4/3}) \ln 2 > (29/40) (69/100) > 1/2$. For $r \geq 2$, we will use
\[\ln \left( 1 + \frac{1}{r} \right) \geq \frac{2}{2r+1}.\]
Therefore, it suffices to verify
\[(1-e^{-4r/3}) \frac{2}{2r+1} \geq \frac1{r+1},\]
which is equivalent to
$e^{4r/3} \geq 2r+2$.
This holds for all $r \geq 2$, since it holds at $r=2$ and the derivative of $e^{4r/3}-2r-2$ is positive for $r \geq 2$. This proves the lemma.
\end{proof}

\section{Properties of Local PAV (Theorem~\ref{thm:local_PAV_properties})}
\label{appendix:Proof_local_PAV_properties}

In this section, we will prove \Cref{thm:local_PAV_properties}, which states that the properties of local PAV rules outlined in \Cref{tab:Local-PAV-Results} hold.

We begin by proving the helper lemma about harmonic numbers that we used in the proof of the $2/3$-PAV approximation guarantee in \Cref{thm:local_PAV_properties} in the main body.
\HarmonicHelperLemma*

\begin{proof}
    If $u = 0$ (and thus $y = 0$), the statement follows trivially since $H(u + x - y) = H(x) \le x$ and $\Gamma = rx$. So, we assume $u > 0$. We get via simple algebraic manipulation that:
\[\Gamma = \frac{r(x - y)}{u + 1} - \frac{y(r - x)}{u(u+1)} = \frac{r(x - y)}{u} - \frac{x(r - y)}{u (u + 1)}.\]
If $x \geq y$, we have:
\begin{align*}
    H(u + x - y) &= H(u) + \sum_{t=u + 1}^{u + x - y} \frac{1}{t}\\
    &\leq H(u) + \frac{x - y}{u + 1}\\
    &= H(u) + \frac{y(r - x)}{r \cdot u(u + 1)} + \frac{\Gamma}{r}\\
    &\leq H(u) + H(u) \frac{r - 1}{2r} + \frac{\Gamma}{r}\\
    &= \frac{3r-1}{2r} H(u) + \frac{\Gamma}{r},
\end{align*}
where the second inequality follows trivially if $x = 0$ (as then $y \leq x = 0$). Otherwise, it follows since $x \geq 1, y \leq u, u \ge 1$ and $H(u) \geq 1$.

If $x < y$, we have:
\begin{align*}
    H(u + x - y) &= H(u) - \sum_{t=u+x-y + 1}^{u} \frac{1}{t}\\
    &\leq H(u) - \frac{y - x}{u}\\
    &= H(u) + \frac{x(r - y)}{r \cdot u(u + 1)} + \frac{\Gamma}{r}\\
    &\leq H(u) + H(u) \frac{r-1}{2r} + \frac{\Gamma}{r}\\
    &= \frac{3r-1}{2r} H(u) + \frac{\Gamma}{r},
\end{align*}
where the second inequality follows similarly to the previous case.
\end{proof}

\LocalPAVProperties*
\begin{proof}
We prove the claims in the order in which they appear in the table, skipping the ones already proved in the main body.

\medskip
\noindent\emph{Pareto optimality.}
For every $\tau \geq 0$, we will show that some \localPAV{$\tau$} committee fails \PO{}. 

Let $k \geq 2$. The set of candidates includes $k$ baseline candidates $w_1,\dots,w_k$ and $k$ auxiliary candidates $a_1,\dots,a_k$. There are two types of voters: (a) for every $j,s \in [k]$, there are $M$ voters approving $\{w_j,a_s\}$, where $M \geq 2$, and (b) for every $j\in[k]$ and every 2-element subset $T \subseteq \{a_1,\dots,a_k\}$, there is one voter approving $\{w_j\}\cup T$.

Let
$W \coloneqq \{w_1,\dots,w_k\}$ and $W' \coloneqq \{a_1,\dots,a_k\}$.
Under $W$, every voter has utility $1$. Under $W'$, voters of the first type still have utility $1$, while voters of the second type have utility $2$. Hence, $W'$ Pareto dominates $W$, so $W$ is not Pareto optimal.

Now consider any swap replacing $w_j$ by $a_t$ in $W$. Voters of the first type contribute
\[
-M(k-1)+\frac{M(k-1)}{2}=-\frac{M(k-1)}{2}.
\]
Indeed, the voters approving $\{w_j,a_s\}$ with $s\ne t$ lose one unit of utility, while the voters approving $\{w_{j'},a_t\}$ with $j'\ne j$ gain from utility $1$ to utility $2$. For voters of the second type, the swap loses one approved candidate for the $\binom{k-1}{2}$ voters in the removed baseline group whose auxiliary pair excludes $a_t$, but gives a $1/2$ marginal PAV gain to $(k-1)^2$ voters in the other groups whose auxiliary pair contains $a_t$, for a net contribution of $-\binom{k-1}{2}+(k-1)^2/2=(k-1)/2$. Thus, the total swap gain is
\[
-\frac{M(k-1)}{2}+\frac{k-1}{2}
=
-\frac{(M-1)(k-1)}{2}<0.
\]
Therefore, $W$ is exact local PAV, yet it is Pareto dominated.

\medskip
\noindent\emph{Weak Pareto optimality.}
We will now argue that the bound on $\tau$ is tight. Consider an instance with $k$ baseline candidates $w_1,\dots,w_k$ and $k$ auxiliary candidates $a_1,\dots,a_k$. For every $j\in[k]$ and every 2-element subset $T\subseteq\{a_1,\dots,a_k\}$, create one voter approving $\{w_j\}\cup T$. Let $W\coloneqq\{w_1,\dots,w_k\}$. Then $n=k\binom{k}{2}$, every voter has utility $1$ under $W$, and the auxiliary committee $W'\coloneqq\{a_1,\dots,a_k\}$ gives every voter utility $2$. Thus, $W$ is not weakly Pareto optimal. For a swap removing $w_j$ and adding $a_t$, the gain in the PAV score is
\[
-\binom{k-1}{2}+\frac{(k-1)^2}{2}
=
\frac{k-1}{2}
=
\frac{n}{k^2}.
\]
Therefore, $W$ is $\tau$-local PAV for every $\tau>n/k^2$, while it is not weakly Pareto optimal.

\medskip
\noindent\emph{\EJRplus{}.}
Suppose $W$ violates \EJRplus{}. Then there exist an unelected candidate $c\notin W$, an integer $\ell\in[k]$, and a group $N'\subseteq N$ such that $|N'| \geq \ell\frac{n}{k}$, every voter in $N'$ approves $c$, and $u_i(W)< \ell$ for every $i \in N'$. Thus, $u_i(W)\le \ell-1$ for every $i\in N'$. 

Consider the $k$ swaps $W \setminus \{a\} \cup \{c\}$, one for each $a\in W$. We will lower-bound the total gain in the PAV score due to these $k$ swaps.

A voter $i\in N'$ contributes at least $(k-\ell+1)/\ell$ over these $k$ swaps. Indeed, if $a\notin A_i$, then the swap increases its utility by one and gives PAV gain $1/(u_i(W)+1)\ge 1/\ell$; this happens for at least $k-(\ell-1)=k-\ell+1$ choices of $a$. Voters who approve $c$ but are not in $N'$ contribute nonnegatively. A voter who does not approve $c$ contributes at least $-1$ in total over the $k$ swaps, because if it has $t>0$ approved candidates in $W$, then removing each of those $t$ candidates causes a loss of $1/t$, and the total loss is exactly $1$. Therefore,
\[
\begin{aligned}
\sum_{a\in W}\bigl(\PAV(W \setminus \{a\} \cup \{c\})-\PAV(W)\bigr)
&\ge |N'|\cdot \frac{k-\ell+1}{\ell}-(n-|N'|) \\
&= |N'|\cdot \frac{k+1}{\ell}-n \\
&\ge \ell\frac{n}{k}\cdot \frac{k+1}{\ell}-n \\
&= \frac{n}{k}.
\end{aligned}
\]
Hence, some swap has gain at least $n/k^2$. Thus, if $\tau\le n/k^2$, this contradicts $\tau$-locality. Therefore, every $\tau$-local PAV committee satisfies \EJRplus{} whenever $\tau\le n/k^2$.

The tightness example is the same as the one for weak \PO{}. In that instance, every auxiliary candidate $a_t$ is approved by $k(k-1)=2n/k$ voters, all of whom have only one approved winner in $W = \{w_1,\dots,w_k\}$. Hence, $W$ violates \EJRplus{} with $\ell=2$. Since every swap gain is exactly $n/k^2$, $W$ is $\tau$-local PAV for every $\tau>n/k^2$.

\medskip
\noindent\emph{$(1+\varepsilon)$-Pareto optimality.}
Let $W$ be a $\tau$-local PAV committee with $\tau\le 1/k^2$. Suppose, for contradiction, that for some $\varepsilon>0$ there is an integral committee $W'$ such that $u_i(W')\ge (1+\varepsilon)u_i(W)$ for every voter $i$, with strict inequality for at least one voter. As before, define $X \coloneqq W'\setminus W$, $Y \coloneqq W\setminus W'$, and $r \coloneqq |X|=|Y|$. Additionally, for any fixed voter $i$, let $u_i\coloneqq |A_i\cap W|$, $x_i\coloneqq |X\cap A_i|$, and $y_i\coloneqq |Y\cap A_i|$.

For every voter with $u_i(W)\ge 1$, integrality implies $x_i\ge y_i+1$, and the same calculation as in the weak \PO{} proof gives $\Gamma_i\ge r/k$. Voters with $u_i(W)=0$ contribute nonnegatively; if such a voter is strictly improved, then it contributes at least $r$. Therefore, the total swap-gain sum is at least $r/k$. On the other hand, $\tau$-locality gives
\[
\sum_{y\in Y}\sum_{x\in X}\Delta_{y,x}<r^2\tau\le \frac{r^2}{k^2}\le \frac{r}{k},
\]
a contradiction. Hence, $W$ is $(1+\varepsilon)$-Pareto optimal for every $\varepsilon>0$.

The following example shows that one cannot extend this guarantee beyond $\tau>1/k$ in general. Let $k\ge 2$. There are candidates $w_1,\dots,w_k$ and $a_1,\dots,a_k$. For each $j\in[k]$, create one voter approving all candidates in $
\{w_1,\dots,w_k\}\setminus\{w_j\}$ and all candidates in $\{a_1,\dots,a_k\}$.

Let $W\coloneqq\{w_1,\dots,w_k\}$ and $W'\coloneqq\{a_1,\dots,a_k\}$. Every voter gets utility $k-1$ from $W$ and utility $k$ from $W'$. Hence, $W'$ $(1+\varepsilon)$-Pareto dominates $W$ for any $0<\varepsilon<1/(k-1)$. For a swap removing $w_j$ and adding $a_t$, only voter $j$ changes utility, increasing from $k-1$ to $k$, so the swap gain is $1/k$. Thus, $W$ is $\tau$-local PAV for every $\tau>1/k$, but it is not $(1+\varepsilon)$-Pareto optimal for $0<\varepsilon<1/(k-1)$.

\par\medskip
\noindent\emph{$2$-\fPO{} counterexamples.}
The factor $2$ can be shown to be asymptotically tight already for exact local PAV due to the following construction: Fix $k\ge 2$. Consider an instance with $k$ baseline candidates $w_1,\dots,w_k$ and $k+1$ auxiliary candidates $a_1,\dots,a_{k+1}$. For every $j\in[k]$ and every 2-element subset $T\subseteq\{a_1,\dots,a_{k+1}\}$, create one voter approving $\{w_j\}\cup T$. Let $W\coloneqq\{w_1,\dots,w_k\}$. Every voter gets utility $1$ from $W$. A swap that removes $w_j$ and adds an auxiliary candidate $a_t$ causes a loss of $\binom{k}{2}$ from group $j$ and a gain of $(k-1)k/2$ from the other groups, so the total swap gain is $0$. Thus, $W$ is exact local PAV. Assigning fractional mass $k/(k+1)$ to each auxiliary candidate gives every voter utility $2k/(k+1)$, while $W$ gives every voter utility $1$. Hence, $W$ is fractionally $\alpha$-Pareto dominated for every $\alpha<2k/(k+1)$, and this tends to $2$.

We will now show that whenever $\tau>1/\lceil k/2\rceil$, some \localPAV{$\tau$} committee may fail to satisfy $2$-\fPO{}. Let $s\coloneqq \left\lfloor\frac{k-1}{2}\right\rfloor$. If $s=0$, consider an instance with a single voter who approves the candidates $a_1,\dots,a_k$, and let $W$ consist only of $k$ unapproved dummy candidates. Every positive swap gain is exactly $1$, so $W$ is $\tau$-local PAV for every $\tau>1=1/\lceil k/2\rceil$, but $W$ is not $2$-\fPO{} because the voter has utility $0$ under $W$ and positive utility under $\{a_1,\dots,a_k\}$.

Now suppose $s\ge 1$. Again consider an instance with a single voter who approves $a_1,\dots,a_k$. Let $W$ contain $s$ approved candidates and $k-s$ unapproved dummy candidates. Any swap that removes a dummy and adds an approved candidate increases the PAV score by $1/(s+1)$, and all other swaps have gain at most this. Hence, $W$ is $\tau$-local PAV for every $\tau>\frac{1}{s+1}=\frac{1}{\lceil k/2\rceil}$.

The committee $W'\coloneqq\{a_1,\dots,a_k\}$ gives the voter utility $k$, while $W$ gives utility $s$. Since $k>2s$, the committee $W'$ (fractionally) $2$-Pareto dominates $W$. Thus, $W$ fails $2$-\fPO{}.

\par\medskip
\noindent\emph{$2/3$-PAV approximation counterexamples.}
The factor $2/3$ is asymptotically tight already for exact local PAV due to the following construction: There are $k$ baseline candidates $w_1,\dots,w_k$ and $k$ auxiliary candidates $a_1,\dots,a_k$. For every pair $(j,t)\in[k]\times[k]$, create two voters who both approve $\{w_j,a_t\}$. Additionally, for every $t\in[k]$, create $k-1$ voters approving exactly $\{a_t\}$. Let $W\coloneqq\{w_1,\dots,w_k\}$. Then $\PAV(W)=2k^2$. A swap replacing $w_j$ by $a_t$ has zero gain: the loss $2(k-1)$ is offset by a gain $k-1$ from pair voters in other groups and a gain $k-1$ from voters approving only $a_t$. Thus, $W$ is exact local PAV. The auxiliary committee $O\coloneqq\{a_1,\dots,a_k\}$ has score
\[
\PAV(O)=2k^2+k(k-1)=3k^2-k.
\]
Hence,
\[
\frac{\PAV(W)}{\PAV(O)}=\frac{2k}{3k-1}\to \frac{2}{3}.
\]
Thus, asymptotically, no multiplicative factor larger than $2/3$ is guaranteed, even for exact local PAV.

Let $\sigma_k \coloneqq \max\left\{s\in\{0,1,\dots,k-1\}: H_s < \frac{2}{3}H_k\right\}$. We will now show that if $\tau > 1/(\sigma_k+1)$, some \localPAV{$\tau$} committee may fail to provide a $2/3$-approximation to the optimal PAV score.

Let $s=\sigma_k$, so $H_s<(2/3)H_k$. Consider a one-voter instance in which the voter approves $a_1,\dots,a_k$. Let $W$ contain $s$ approved candidates and $k-s$ unapproved dummy candidates. Any swap that removes a dummy and adds an approved candidate increases the PAV score by $1/(s+1)$, and all other swaps have gain at most this. Thus, $W$ is $\tau$-local PAV for every $\tau>1/(s+1)=1/(\sigma_k+1)$. However, $\PAV(W) = H_s$, which is strictly less than $2/3$ of the maximum possible PAV score of $H_k$, achieved by the committee $a_1,\dots,a_k$. This completes the proof of \Cref{thm:local_PAV_properties}.
\end{proof}

\subsection{Justifying Parameter-Dependent Bounds}
\label{appendix:Justifying_Instance-Dependent_Bounds}

In \Cref{rem:Instance-Dependence-Local-PAV} regarding \localPAV{$\tau$}, we stated that parameter-independent bounds on $\tau$ are often overly pessimistic. We will provide justification for this statement in the proposition below.

\begin{proposition}
Fix any constant $\tau>0$. If the committee size $k$ is allowed to grow, then \localPAV{$\tau$} admits no $k$-independent approximation factor for fractional Pareto optimality, and admits no positive $k$-independent multiplicative approximation to the optimal PAV score.
\label{prop:One_Voter_Example_Instance-Dependent_Justification}
\end{proposition}

\begin{proof}
Let $s\coloneqq \left\lceil \frac{1}{\tau}\right\rceil$. Then $s$ is a constant depending only on $\tau$, and $\frac{1}{s+1}<\tau$. Consider any committee size $k>s$. There is a single voter. The candidate set consists of approved candidates $a_1,\ldots,a_k$ and dummy candidates $d_1,\ldots,d_k$. The voter approves exactly the candidates $a_1,\ldots,a_k$ and approves no dummy candidate. Consider the committee
\[
W\coloneqq \{a_1,\ldots,a_s\}\cup \{d_1,\ldots,d_{k-s}\}.
\]
Thus $|W|=k$ and the voter's utility under $W$ is $s$.

We first verify that $W$ is \localPAV{$\tau$}. Consider any swap that removes a candidate from $W$ and adds a candidate outside $W$. Since there is only one voter, the PAV score is simply $H_u$, where $u$ is the number of approved candidates in the committee. If the swap removes a dummy and adds an unselected approved candidate, the utility increases from $s$ to $s+1$, so the PAV gain is
\[
H_{s+1}-H_s=\frac{1}{s+1}<\tau.
\]
If the swap removes an approved candidate and adds another approved candidate, or removes a dummy and adds another dummy, the utility is unchanged and the gain is $0<\tau$. Finally, if the swap removes an approved candidate and adds a dummy, the utility decreases from $s$ to $s-1$, so the gain is negative. Hence every single-swap PAV gain is strictly smaller than $\tau$, and $W$ is \localPAV{$\tau$}.

Now consider the all-approved committee $O\coloneqq \{a_1,\ldots,a_k\}$. This committee gives the voter utility $k$. Hence, for any fixed $\alpha$, if $k>\alpha s$, then $O$ itself, viewed as a fractional committee, gives utility $k>\alpha s=\alpha\,u(W)$. Thus, $W$ is $\alpha$-fractionally Pareto dominated. Since $s$ is fixed while $k$ can be arbitrarily large, no finite $k$-independent $\fPO$ approximation factor can be guaranteed by \localPAV{$\tau$}.

The same instance also rules out any positive $k$-independent PAV-score approximation. The PAV score of $W$ is $\PAV(W)=H_s$, whereas the optimal PAV score is achieved by $O$ and equals $\PAV(O)=H_k$. Therefore, the approximation ratio of $W$ is
\[
\frac{\PAV(W)}{\PAV(O)}=\frac{H_s}{H_k}.
\]
Since $s$ is fixed as a function of $\tau$ and $H_k\to\infty$ as $k\to\infty$, this ratio tends to $0$. Consequently, for every constant $\rho>0$, choosing $k$ sufficiently large gives a \localPAV{$\tau$} committee with
\[
\frac{\PAV(W)}{\max_{|S|=k}\PAV(S)}<\rho.
\]
Thus, no positive $k$-independent multiplicative approximation to the optimal PAV score can be guaranteed.
\end{proof}

\section{Proof of Asymptotic Tightness of $\alpha^\star$-\fPO{} for PAV}
\label{appendix:PAV_fPO_LowerBound}

We will show that the $\alpha^\star$-\fPO{} guarantee for PAV is asymptotically tight. Towards proving this, we will first establish a useful inequality in \Cref{lem:Calculus_Inequality}.

\begin{lemma}
Let $\phi(z)\coloneqq \frac{1-e^{-z}}{z}$, with $\phi(0)=1$. For every $0<\gamma<\alpha^\star$ and every
$\lambda\in(0,1]$,
\[
    g_\gamma(\lambda)
    \coloneqq h(\gamma\lambda)+(1-\lambda)\phi(\gamma\lambda)
    <1.
\]
\label{lem:Calculus_Inequality}
\end{lemma}
\begin{proof}
We will first prove the weak inequality at $\gamma=\alpha^\star$. Let $a=\alpha^\star$ and,
for $\mu\in[0,a]$, define
\[
    q(\mu)\coloneqq 1-h(\mu)-\left(1-\frac{\mu}{a}\right)\phi(\mu).
\]
We show $q(\mu)\geq 0$. Clearly $q(0)=q(a)=0$. For $\mu\in(0,a)$, direct
differentiation gives
\[
    q'(\mu)
    =\frac{\mu^2e^{-\mu}-a(\mu-1+e^{-\mu})}{a\mu^2}.
\]
Since $\mu-1+e^{-\mu}>0$ for $\mu>0$, the sign of $q'(\mu)$ is the sign of
$R(\mu)-a$, where
\[
    R(\mu)\coloneqq \frac{\mu^2e^{-\mu}}{\mu-1+e^{-\mu}}.
\]
The function $R$ is strictly decreasing on $(0,\infty)$, because
\[
    R'(\mu)
    =-\frac{\mu\left(e^\mu(\mu^2-2\mu+2)-2\right)}{
      \left(\mu e^\mu-e^\mu+1\right)^2}<0;
\]
the numerator is positive for $\mu>0$ since
$e^\mu(\mu^2-2\mu+2)-2$ has derivative $e^\mu\mu^2>0$ and value $0$ at
$\mu=0$. Moreover, $\lim_{\mu\downarrow 0}R(\mu)=2>a$, while $R(a)<a$ because
$a>1$, $R(1)=1$, and $R$ is decreasing. Hence $q'$ is positive up to one point and
negative after that point. Since $q(0)=q(a)=0$, this implies $q(\mu)\geq 0$
for all $\mu\in[0,a]$. Equivalently,
\[
    h(\mu)+\left(1-\frac{\mu}{a}\right)\phi(\mu)\leq 1
    \qquad\text{for every } \mu\in[0,a].
\]
Now fix $\gamma<a$ and $\lambda\in(0,1]$, and set $\mu=\gamma\lambda$. Then
$\lambda=\mu/\gamma>\mu/a$, and $\phi(\mu)>0$, so
\[
    g_\gamma(\lambda)
    =h(\mu)+(1-\lambda)\phi(\mu)
    < h(\mu)+\left(1-\frac{\mu}{a}\right)\phi(\mu)
    \leq 1.\qedhere
\]
\end{proof}

We will now establish the tightness guarantee in \Cref{thm:global_PAV_properties}. Specifically, we will construct, for every $\alpha<\alpha^\star$, an instance where a PAV committee is fractionally $\alpha$-Pareto dominated.

\begin{proof} (Tightness for \Cref{thm:global_PAV_properties})
Fix $\alpha<\alpha^\star$ and choose a rational number $\gamma$ such that $\alpha<\gamma<\alpha^\star$. Write $\gamma=d/c$ for positive integers $c,d$. By replacing $c,d$ by a sufficiently large common multiple, we may assume that they are as large as needed.

For sufficiently large $k$, construct the following instance. There are $k$ baseline candidates $w_1,\dots,w_k$ and $ck$ auxiliary candidates $a_1,\dots,a_{ck}$. For every $j\in[k]$ and every $d$-element subset $T\subseteq\{a_1,\dots,a_{ck}\}$, create one voter approving exactly $\{w_j\}\cup T$. 

Let $W\coloneqq \{w_1,\dots,w_k\}$. Every voter gets utility $1$ from $W$, so $\PAV(W)=n$. Now define a fractional committee $x$ by assigning mass $1/c$ to each auxiliary candidate and mass $0$ to each baseline candidate. This is feasible because there are $ck$ auxiliary candidates. Every voter approves exactly $d$ auxiliary candidates, so every voter receives fractional utility
\[
u_i(x)=\frac{d}{c}=\gamma>\alpha.
\]
Thus, $x$ fractionally $\alpha$-Pareto dominates $W$. It remains to show that, for the parameters chosen above, $W$ is a global PAV committee.

Consider an arbitrary committee $S$ of size $k$. Suppose $S$ contains $\ell\ge 1$ auxiliary candidates, and write $\lambda\coloneqq \frac{\ell}{k}$. Then $S$ contains $k-\ell$ baseline candidates. By symmetry, the normalized PAV score of $S$ depends only on $\ell$; here, the normalized score of $S$ refers to its PAV score as a fraction of that of a PAV committee. 

Let $Y\sim \operatorname{Hypergeom}(ck,\ell,d)$, namely, $Y$ is the number of chosen auxiliary candidates contained in a uniformly random $d$-subset of the $ck$ auxiliary candidates. A $1-\lambda$ fraction of the baseline voter groups keep their baseline candidate and a $\lambda$ fraction lose it. Therefore,
\[
\frac{\PAV(S)}{n}
=
\varphi_{k,c,d}(\lambda)
\coloneqq
(1-\lambda)\mathbb{E}[H_{1+Y}]
+
\lambda\mathbb{E}[H_Y].
\]

We claim that $\varphi_{k,c,d}(\lambda)<1$ for every $\lambda\in\{1/k,2/k,\dots,1\}$, once $c,d$ and then $k$ are chosen sufficiently large.

First, let's consider small values of $\lambda$. Since $\mathbb{E}[Y]=\gamma\lambda$, and since $H_y\le y$ and $H_{1+y}\le 1+y/2$ for every integer $y\ge 0$,
\[
\begin{aligned}
\varphi_{k,c,d}(\lambda)
&\le (1-\lambda)\left(1+\frac{\gamma\lambda}{2}\right)+\lambda\gamma\lambda \\
&=1-\left(1-\frac{\gamma}{2}\right)\lambda+\frac{\gamma}{2}\lambda^2.
\end{aligned}
\]
Since $\gamma<\alpha^\star<2$, there is $\delta>0$ such that $\varphi_{k,c,d}(\lambda)<1$ for all $0<\lambda\le\delta$, uniformly over $k,c,d$ with $d/c=\gamma$.

It remains to handle $\lambda\in[\delta,1]$. For fixed $c,d$, as $k\to\infty$, the hypergeometric random variable $Y$ converges uniformly in $\lambda\in[\delta,1]$ to $Z\sim \operatorname{Bin}\left(d,\frac{\lambda}{c}\right)$. Then, as $c,d\to\infty$ while keeping $d/c=\gamma$, this binomial random variable converges uniformly in $\lambda\in[\delta,1]$ to $P_\lambda\sim \operatorname{Pois}(\gamma\lambda)$. For $t=\gamma\lambda$, $\mathbb{E}[H_{P_\lambda}]=h(t)$, and
$\mathbb{E}[H_{1+P_\lambda}]
= h(t)+\mathbb{E}\left[\frac{1}{1+P_\lambda}\right]
= h(t)+\frac{1-e^{-t}}{t}$. Thus, the limiting normalized score is
\[
g_\gamma(\lambda)
\coloneqq
h(\gamma\lambda)
+
(1-\lambda)\frac{1-e^{-\gamma\lambda}}{\gamma\lambda}.
\]
From \Cref{lem:Calculus_Inequality}, we know that if $\gamma<\alpha^\star$, then $g_\gamma(\lambda)<1$ for every $\lambda\in(0,1]$. By compactness, $g_\gamma$ is bounded away from $1$ on $[\delta,1]$. The uniform convergences above therefore imply that, after choosing $c,d$ sufficiently large and then $k$ sufficiently large, we have $\varphi_{k,c,d}(\lambda)<1$ for every $\lambda\in[\delta,1]\cap\{1/k,2/k,\dots,1\}$. Together with the small-$\lambda$ bound, this shows that every committee using at least one auxiliary candidate has PAV score strictly smaller than $n=\PAV(W)$. The only size-$k$ committee using no auxiliary candidates is $W$ itself. Hence, $W$ is the unique global PAV committee.

We have constructed, for every $\alpha<\alpha^\star$, an instance with a global PAV committee $W$ that is fractionally $\alpha$-Pareto dominated. Thus, the factor $\alpha^\star$ is asymptotically tight.
\end{proof}

\section{Properties of the DMMS algorithm}
\label{appendix:DMMS_Fails_JR}
The DMMS  algorithm~\citep{DMM+21tight} considers an LP relaxation obtained by linearly interpolating the harmonic scores. Formally, let $\widetilde{H}(t) = H(\lfloor t \rfloor) + \frac{\{t\}}{\lfloor t \rfloor + 1}$. This function $\widetilde{H}$ is a continuous, concave piecewise-linear approximation of the harmonic numbers $H_t$. They maximize $\sum_{i \in N} \widetilde{H}(u_i(x))$ over the fractional committee polytope $P_k \coloneqq \{x \in [0,1]^C : \sum_{c \in C} x_c = k\}$. Then, they round the fractional solution using randomized pipage rounding. We first prove the output of their algorithm satisfies \fPO{}.

\begin{lemma} \label{lem:dmms_fPO}
    The DMMS algorithm satisfies \fPO{}.
\end{lemma}

\begin{proof}
~\citet{BBF+26fractional} show that a committee $W$ is \fPO{} if and only if it is possible to assign a positive weight $w_i > 0$ to each voter $i \in N$ such that $W$ maximizes the sum of weighted utilities of the voters, among all committees of size $k$. We will first prove that the DMMS algorithm satisfies \fPO{} by producing positive voter weights that satisfy this condition.

Consider the fractional committee $x^\star$ obtained in their first step, that is the optimal solution of the following concave maximization problem:
\[
\max_{x \in P_k} \sum_{i \in N} \widetilde{H}(u_i(x)),
\]
where $P_k = \{x \in [0,1]^C : \sum_{c \in C} x_c = k\}$ is the fractional committee polytope. By the Karush--Kuhn--Tucker (KKT) first-order optimality conditions for concave maximization over $P_k$, there exist subgradient weights $w_i \in \partial \widetilde{H}(u_i(x^\star))$ for each voter $i \in N$ and a Lagrange multiplier $\lambda \in \mathbb{R}$ corresponding to the budget constraint $\sum_{c \in C} x_c = k$. Since $\widetilde{H}$ is strictly increasing with positive slopes everywhere (the left- and right-derivatives satisfy $\widetilde{H}'(u) > 0$), every subgradient weight is strictly positive: $w_i > 0$ for all $i \in N$. For each candidate $c \in C$, define the weighted candidate score
\[
q_c \coloneqq \sum_{i \in N :\, c \in A_i} w_i.
\]
The KKT conditions imply that $x^\star$ maximizes the linear objective $\sum_{c \in C} q_c x_c$ over $P_k$. In particular, the complementary slackness conditions ensure that:
\[
q_c \ge \lambda \ \text{ if } x^\star_c = 1, \qquad\qquad
q_c = \lambda \ \text{ if } x^\star_c \in (0,1), \qquad\qquad
q_c \le \lambda \ \text{ if } x^\star_c = 0.
\]
In the rounding phase, pipage rounding operates exclusively on pairs of fractional coordinates $x_c \in (0,1)$, shifting mass between them until coordinates become integral ($0$ or $1$). Crucially, any coordinate with $x^\star_c \in \{0, 1\}$ is never modified. Thus, pipage rounding ensures that the final integral committee $W_0$ satisfies $\{c \in W_0: x^\star_c = 1\} \subseteq W_0 \subseteq \{c \in C: x^\star_c > 0\}$. Thus, $q_c \geq \lambda$ for every $c \in W_0$, and $q_c \leq \lambda$ for every $c \notin W_0$. Hence, the output of the DMMS algorithm satisfies \fPO{}.
\end{proof}

We now show that the DMMS algorithm might not satisfy \JR{}. One natural attempt to ensure \JR{} is to perform local search (PAV score improving swaps) starting from their output committee. While this modified algorithm satisfies \JR{}, we show that it might not satisfy better than $2$-\fPO{}.

\begin{lemma}\label{lem:dmms_not_JR}
    The DMMS algorithm can fail to satisfy \JR{}. Furthermore, if the algorithm is modified by performing local-search to obtain \JR{} in the end, the output might not satisfy better than $2$-\fPO{}.
\end{lemma}
\begin{proof}
Fix $k \ge 3$. We consider an election instance with candidate set $C = L \cup S \cup R$, where
\[
L = \{\ell_1, \dots, \ell_k\}, \qquad
S = \{s_1, \dots, s_k\}, \qquad
R = \{r_1, \dots, r_{k+1}\}.
\]
There are two types of voters:
\begin{enumerate}
    \item For each $i \in [k]$ and each pair $\{p, q\} \in \binom{[k+1]}{2}$, there are $4$ voters who approve $\{\ell_i, r_p, r_q\}$. The number of such voters is $4k\binom{k+1}{2} = 2k^2(k+1)$.
    \item For each $j \in [k]$ (with indices modulo $k$), there are $k^2$ voters who approve $\{s_j, s_{j+1}\}$. The number of such voters is $k \cdot k^2 = k^3$.
\end{enumerate}
The total number of voters is $n = 4k\binom{k+1}{2} + k^3 = k^2(3k+2)$. In this instance:
\begin{itemize}
    \item Each candidate $\ell_i \in L$ is approved by $4\binom{k+1}{2} = 2k(k+1)$ first-type voters.
    \item Each candidate $r_p \in R$ is approved by $4 \cdot k \cdot k = 4k^2$ first-type voters.
    \item Each candidate $s_j \in S$ is approved by $2k^2$ second-type voters.
\end{itemize}

Consider the fractional committee $x^\star \in [0,1]^C$ given by
\[
    x^\star_{\ell_i} = 0 \quad (\forall i \in [k]), \qquad
x^\star_{r_j} = \frac{3k+1}{4(k+1)} \quad (\forall j \in [k+1]), \qquad
x^\star_{s_p} = \frac{k-1}{4k} \quad (\forall p \in [k]).
\]
Note that $\sum_{c \in C} x^\star_c = (k+1)\frac{3k+1}{4(k+1)} + k\frac{k-1}{4k} = k$, so $x^\star$ is feasible.

Under $x^\star$, each first-type voter approving $\{\ell_i, r_p, r_q\}$ has fractional utility $2x^\star_{r_j} = \frac{3k+1}{2(k+1)} \in (1,2)$, so the slope of the objective for every first-type voter is $1/2$. Each second-type voter approving $\{s_j, s_{j+1}\}$ has fractional utility $2x^\star_{s_j} = \frac{k-1}{2k} < 1$, so their slope is $1$. The candidate marginal scores are:
\[
    q_{\ell_i} = 2k(k+1) \cdot \frac{1}{2} = k(k+1) < 2k^2, \qquad
q_{r_j} = 4k^2 \cdot \frac{1}{2} = 2k^2, \qquad
q_{s_j} = 2k^2 \cdot 1 = 2k^2.
\]
Because every candidate in $R \cup S$ achieves the maximum score $2k^2$ and every candidate in $L$ has strictly smaller score, $x^\star$ is an optimal solution to the DMMS LP.

There exists a valid execution of randomized pipage rounding on $x^\star$ that repeatedly pairs a fractional coordinate in $S$ with a fractional coordinate in $R$ and shifts mass to increase the coordinate in $S$. At the end of this process, all coordinates in $S$ reach $1$ and all coordinates in $R$ reach $0$, yielding the integral committee
\[
W_0 = S = \{s_1, \dots, s_k\}.
\]
Under $W_0 = S$, none of the first-type voters are represented (each receives utility $0$). For any candidate $r_j \in R$, the number of first-type voters approving $r_j$ is $4k^2$. The \JR{} representation quota is
\[
\frac{n}{k} = \frac{k^2(3k+2)}{k} = 3k^2 + 2k.
\]
Since $4k^2 \ge 3k^2 + 2k$ for every $k \ge 3$, the supporters of $r_j$ form a cohesive group of size at least $n/k$ with zero representation in $W_0$. Hence, $W_0$ violates \JR{}.

Now consider modifying the DMMS algorithm by performing local search via PAV-improving swaps until no further improvement is possible. Starting from $W_0 = S$, consider replacing the candidates in $S$ one by one with candidates from $L$. When a candidate $\ell_i \in L$ is added and a candidate $s_j \in S$ is removed:
\begin{itemize}
    \item The $4\binom{k+1}{2} = 2k(k+1)$ first-type voters approving $\ell_i$ previously had utility $0$ and now have utility $1$, contributing a gain of $2k(k+1)$ to the PAV score.
    \item At most $2k^2$ second-type voters approve $s_j$, and each loses at most $1$ in PAV score.
\end{itemize}
Thus, the net change in PAV score for each such swap satisfies
\[
\Delta \PAV \ge 2k(k+1) - 2k^2 = 2k > 0.
\]
(Note that $2k > \frac{3k+2}{2} = \frac{n}{2k^2} \ge \frac{1}{2k^2}$ for $k \ge 3$, so these swaps are strictly improving even under thresholded local search. Furthermore, by creating two copies of the instance and doubling $k$, we can make sure that the swaps increase PAV score by at least $\frac{n}{k^2}$, a threshold necessary to ensure that the final committee satisfies \JR{}.) After $k$ such swaps, the algorithm reaches the committee
\[
W = L = \{\ell_1, \dots, \ell_k\}.
\]

We verify that $W = L$ is locally optimal under single swaps:
\begin{itemize}
    \item For a swap replacing $\ell_i \in L$ with $r_j \in R$: Removing $\ell_i$ causes a loss of $1$ from each of the $4\binom{k}{2} = 2k(k-1)$ voters approving $\{\ell_i, r_p, r_q\}$ with $p, q \ne j$ (whose utility drops from $1$ to $0$). Adding $r_j$ increases the utility of the $4k$ voters approving $\{\ell_{i'}, r_j, r_p\}$ ($p \ne j$) from $1$ to $2$ for each $i' \ne i$, contributing a gain of $(k-1) \cdot \left(4k \cdot \frac{1}{2}\right) = 2k(k-1)$. The remaining voters either do not approve $\ell_i$ or $r_j$, or approve both (keeping utility $1$). Hence, the net PAV change is $-2k(k-1) + 2k(k-1) = 0$.
    \item For a swap replacing $\ell_i \in L$ with $s_j \in S$: Adding $s_j$ gains at most $2k^2$ from second-type voters, while removing $\ell_i$ loses $2k(k+1)$. Thus, $\Delta \PAV \le 2k^2 - 2k(k+1) = -2k < 0$.
\end{itemize}
Thus, $W = L$ is a local PAV committee. 

Finally, we evaluate the fractional Pareto optimality of $W = L$. Consider the fractional committee $y \in [0,1]^C$ defined by
\[
y_{r_j} = \frac{k}{k+1} \quad (\forall j \in [k+1]), \qquad
y_c = 0 \quad (\forall c \notin R),
\]
which is feasible since $\sum_{j=1}^{k+1} y_{r_j} = k$. Under $W = L$ and $y$:
\begin{itemize}
    \item Every first-type voter approves $\ell_i \in W$, so $u_v(W) = 1$. Under $y$, they approve $r_p$ and $r_q$, so $u_v(y) = y_{r_p} + y_{r_q} = \frac{2k}{k+1}$.
    \item Every second-type voter receives utility $0$ under both $W$ and $y$.
\end{itemize}
Consequently, for every $\beta < \frac{2k}{k+1}$, the fractional committee $y$ $\beta$-dominates $W$. Thus, $W$ is not $\beta$-\fPO{} for any $\beta < \frac{2k}{k+1}$. Since $\lim_{k \to \infty} \frac{2k}{k+1} = 2$, the output committee does not satisfy better than $2$-\fPO{}.
\end{proof}

\section{Approaches That Fail To Show Existence of \JR{} and \fPO{}}
\label{appendix:Limitations}

Demonstrating the existence of committees that simultaneously satisfy \JR{} and \fPO{} was a folklore open problem for some time. It was very recently resolved negatively: \citet[Theorem 1]{jin2026limitationsbobw} present an example where every \JR{} committee is dominated by a fractional committee (indeed, even \emph{strictly} dominated; every voter is strictly better off under the fractional committee). Before this example was discovered, we attempted to prove existence using a variety of natural proof strategies. As we now know, all of them were doomed to fail. We list explicit counterexamples for each of them, in the hope that they will be useful to guide future research, for example by providing test examples and barriers to algorithms that work in special cases or for small parameter values, or that relax \JR{} in some way.

First, let us better understand the conditions under which a \JR{}+\fPO{} committee must always exist.~\citet{BBF+26fractional} show that a committee $W$ is \fPO{} if and only if it is possible to assign a positive weight $w_i > 0$ to each voter $i \in N$ such that $W$ maximizes the sum of weighted utilities of the voters, among all committees of size $k$. Let us define the \emph{weighted support} of a candidate as the sum of the weights of the voters who approve it. Thus, for each \fPO{} committee $W$, there exist positive voter weights $(w_i)_{i\in N}$ such that $W$ consists of the $k$ candidates with the highest weighted support, breaking ties arbitrarily. Note that in general the weighted support of the top $k$ candidates might not be the same. However, we first argue that, for a \JR{}+\fPO{} committee to always exist, it suffices to look at the ``top set'', that is the set of candidates who have the highest weighted support (and every other candidate has strictly smaller weighted support).

\begin{definition}[Top set]
    A set $W$ of candidates is said to be a \emph{top set}, if it is possible to assign positive weights to the voters, such that $W$ is the set of the candidates with the highest weighted support.
\end{definition}

Note that the definition of the top set is independent of $k$. In particular, a top set might have size smaller than $k$. Nevertheless, we will argue that if \JR{}+\fPO{} exists for all instances, it is always sufficient to select a subset of some top set, and then add candidates with the next highest weighted supports if needed, to make the size exactly equal to $k$. Crucially, we argue that the \JR{} property is already satisfied even before adding the extra candidates.
\begin{lemma}
    A $\JR{}+\fPO{}$ committee always exists if and only if, for each instance $\I = \langle N, C, \A, k \rangle$, there exists a committee $W$ of size at most $k$ which is a subset of some top set and for each candidate $c$, the number of voters who approve $c$ but do not approve anyone in $W$ is strictly less than $n/k$.
\end{lemma}
\begin{proof}
    If there exists such a committee $W$, then clearly \JR{} + \fPO{} exists, because we can begin with the committee $W$, and add the $k-|W|$ candidates with the next highest scores, breaking ties arbitrarily. Now, suppose \JR{} + \fPO{} exists for all instances. Consider any instance $\I = \langle N, C, \A, k \rangle$. Consider a new instance $\I'$, where the voters remain the same but each candidate from $C$ is copied $k$ times. Consider any \JR{}+\fPO{} committee $W'$ in this instance. Clearly, since $W'$ is \fPO{}, there exist voter weights $w$, such that $W'$ consists of the $k$ candidates with the highest weighted support, breaking ties arbitrarily. Additionally, the weighted support of these $k$ candidates must be exactly the same, since there were $k$ copies of each candidate. This is because if a candidate $c_1 \in W'$ had a strictly larger weighted support than some $c_2 \in W'$, then all the $k$ copies of $c_1$ have a higher weighted support than $c_2$, and $c_2$ could not have been in the top $k$ candidates. Now, let $W$ denote the set of candidates in the original instance who have at least one copy in $W'$. It is easy to see that $W$ satisfies the required properties for the lemma statement, using the same voter weights $w$.
\end{proof}

\subsection{Sequential Weight-Raising Algorithm} 
A major success story of using \fPO{} was the discovery of a pseudo-polynomial time algorithm for finding an allocation of indivisible goods that is \PO{} and EF1 \citep{BKV18ef1}. While that algorithm was originally phrased using Fisher markets terminology, one can also view it as an algorithm that begins with the weight vector $w = (1, \dots, 1)$ and an allocation that maximizes weighted welfare, and that then repeatedly increases the weight of the worst-off agent and updates the allocation until a fair outcome is achieved. 

The same idea can be ported to voting settings. For example, in the setting of ``fair mixing'', which finds a distribution over candidates based on approval preferences, \citet{BBP+21distribution} developed the ``sequential utilitarian rule'' that repeatedly increases voter weights until an outcome is obtained where each voter approves at least an $\frac1n$ portion of the outcome (the same rule had been discussed earlier in a working paper version of \citet{BMS05}). A similar idea is used by the ``redistributive utilitarian rule'' proposed by \citet{SV24maximumflowfairnetwork} for selecting a fractional committee.

For the setting of (non-fractional) approval-based committee elections, the natural analog would be the following algorithm: Set $w = (1, \dots, 1)$, and consider the top set $T$ of candidates with the highest score. Find all voters who don't approve any candidate in $T$, and continuously increase their weights at the same rate, until a new candidate joins $T$. Then continue increasing weights for the voters who still do not approve any candidate in $T$, and so on. This procedure terminates with a top set $T$ such that $T \cap A_i \neq \emptyset$ for every $i \in N$. Thus, our task now reduces to finding a committee $W \subseteq T$ that satisfies \JR{}. Unfortunately, there are profiles where no such \JR{} committee exists.

\begin{example}[Weight-Raising Algorithm Fails \JR{}]
	
    Take $n=39$ voters, $m = 24$ candidates, $k = 3$. The voters are partitioned into $A = \{a_1, \ldots, a_{17}\}$, $B = \{b_1, \ldots, b_{16}\}$, and $D = \{d_1, \ldots, d_6\}$. There is one candidate approved by $A$, one approved by $B$, one approved by $A \setminus \{a_i\} \cup \{b_i\}$ for each $i \in [16]$, and one approved by $A \setminus \{a_i\} \cup \{d_i\}$ for each $i \in [6]$. Let the candidate approved by $B$ be candidate $c$. Let $C$ be the set of all the candidates.

    Note that, when all the voter weights are $1$, $c$ has a weighted support of $16$, while every other candidate has a weighted support of $17$. Hence, the top set initially is $C \setminus \{c\}$ and each voter approves at least one candidate in $C \setminus \{c\}$. However, any subset of size $3$ of $C \setminus \{c\}$ can cover at most $3$ voters from $B$, leaving $16 - 3 = 13 = \frac{n}{k}$ voters in $B$ unrepresented. All these voters approve $c$, and hence there is no \JR{} committee within $C \setminus \{c\}$. 
\end{example}

Nevertheless, the above approach proves the existence of \JR{}+\fPO{} whenever $k \geq n - 1$. If $k \ge n$, we can ask each voter to select one candidate that they approve from the final top set (recall that each voter approves at least one candidate from the final top set) and fill up any remaining seats with arbitrary highest-score candidates (according to the final voter weights). If $k = n - 1$, we can select any $n - 1$ voters to select one candidate each that they approve. This makes sure that at least $n - 1$ voters are represented, and hence the resulting committee satisfies \JR{} since $n/k > 1$.

\subsection{Sufficient Conditions on the Weight Vector}

The algorithm for finding an EF1 and \fPO{} allocation of indivisible goods \citep{BKV18ef1,GM24fPO,M25polynomial} uses the concept of ``price EF1'' (pEF1), which can be seen as a sufficient condition on the weight vector $w$ and the allocation $A$ maximizing $w$-weighted welfare such that $A$ will satisfy EF1. In particular, without introducing the specific notation of the indivisible goods setting, let us write down the definition of pEF1 in its weight form rather than the standard form in terms of Fisher market prices (for notational clarity we assume that $A_i \neq \emptyset$ and $u_i(A_i) > 0$ but the condition can be extended to cover those cases):
\[
	\frac{1}{u_i(A_i)} \le w_i \le \frac{1}{u_i(A_i \setminus \{o_i\})} \quad\text{for each $i\in N$ and some $o_i \in A_i$.}
\]
This condition says that the weight $w_i$ of each agent should be approximately equal to $1/u_i(A_i)$, where $u_i(A_i)$ is the utility that agent $i$ obtains in outcome $A$. Interestingly, ``$w_i = \frac{1}{u_i}$'' also appears in the first-order condition of maximizing Nash welfare for allocating divisible private goods.

This discussion motivates the search for a similar sufficient condition which will imply \JR{}, i.e., a kind of ``price-JR''. A natural place to look for such a condition is the proof that PAV satisfies \EJRplus{} \citep{ABC+17justified}, which involves the familiar-looking terms ``$1/u_i(W)$''. Indeed, a sufficient condition for satisfying \EJRplus{} based on this idea can be formulated as follows.

\begin{lemma}
	\label{lem:sufficient-condition-reciprocal-utility}
	Consider positive voter weights $w$ and a committee $W$ maximizing $w$-weighted welfare such that 
	\[
		\frac{1}{u_i(W)  + 1} \le w_i < \frac{1}{u_i(W)} \quad\text{for each $i\in N$.}
	\]
	Here, $1/u_i(W)$ is interpreted as $+\infty$ when $u_i(W)=0$. Then $W$ satisfies \EJRplus{}.
\end{lemma}
\begin{proof}
	Consider the total weighted support of $W$, which satisfies
	\[
		\sum_{c \in W} \sum_{i \in \supp(c)} w_i
		< \sum_{c \in W} \sum_{i \in \supp(c)} \frac{1}{u_i(W)}
		= \sum_{i \in N : u_i(W) > 0} \sum_{c \in A_i \cap W} \frac{1}{u_i(W)}
		\le n.
	\]
	Hence, there exists some $c^\dagger \in W$ with weighted support strictly less than $\frac{n}{k}$.
	
	Suppose that $W$ fails \EJRplus{}, because some group $S \subseteq N$ with $|S| \ge  \ell \frac{n}{k}$ all have utility at most $\ell - 1$ and all approve $d \in C \setminus W$. Then the $w$-score of $d$ satisfies
	\[
		\sum_{i \in \supp(d)} w_i 
		\ge 
		\sum_{i \in S} w_i 
		\ge
		\sum_{i \in S} \frac{1}{u_i(W)  + 1}
		\ge
		\sum_{i \in S} \frac{1}{(\ell - 1) + 1}
		=
		\frac{|S|}{\ell}
		\ge
		\frac{n}{k}.
	\]
	Hence, $d$ has strictly higher $w$-score than $c^\dagger$, but $d \not\in W$ while $c^\dagger \in W$. This contradicts that $W$ maximizes $w$-weighted welfare.
\end{proof}

If we aim for just \JR{} instead of \EJRplus{}, we can find a weaker sufficient condition that still makes the proof of \Cref{lem:sufficient-condition-reciprocal-utility} go through:
\begin{lemma}
	\label{lem:weaker-sufficient-condition}
	Consider positive voter weights $w$ and a committee $W$ maximizing $w$-weighted welfare such that every voter $i \in N$ with $u_i(W) = 0$ has $w_i \ge 1$ and the total weighted support of $W$ satisfies $\sum_{c \in W} \sum_{i \in \supp(c)} w_i < n$. Then $W$ satisfies \JR{}.
\end{lemma}
This sufficient condition suggests the following rule: minimize, over all positive weight vectors $w$ and committees $W$ that maximize $w$-weighted welfare, the total weighted support of $W$, subject to $w_i \ge 1$ for all voters $i$ with $u_i(W) = 0$. If we could prove that the optimum objective value of this optimization is always strictly below $n$, then we would have found a rule producing \JR{} committees. Unfortunately, the rule fails.

\begin{example}[Sufficient condition on weights fails]
	Let $n = 6$, $m = 6$, $k = 3$. The approval sets are $A_1 = \{a,b\}$, $A_2 = \{a,c\}$, $A_3 = \{b,c\}$, $A_4 = \{d,e\}$, $A_5 = \{d,f\}$, $A_6 = \{e,f\}$. 
    
    The profile consists of two disjoint triangles. Consider any positive weight vector $w$ and any $w$-weighted-welfare maximizing committee $W$ satisfying $w_i\ge 1$ for every unrepresented voter. If $W$ chooses all three candidates from one triangle, then all three voters in the other triangle are unrepresented, so each unelected candidate has weighted support at least $2$. Since every elected candidate must have weighted support at least that of every unelected candidate, the total weighted support of $W$ is at least $3\cdot 2=6=n$.
	
	The remaining case is that $W$ chooses two candidates from one triangle and one from the other. By symmetry, suppose $W=\{a,b,d\}$, so voter $6$ is unrepresented and $w_6\ge 1$. Since $d$ is elected while $e$ and $f$ are not, we must have $w_4+w_5\ge w_4+w_6$ and $w_4+w_5\ge w_5+w_6$, hence $w_4,w_5\ge w_6\ge 1$. Thus, $d$, $e$, and $f$ all have weighted support at least $2$. Since $a$ and $b$ are elected, they also have weighted support at least $2$. Hence, the total weighted support of $W$ is again at least $6=n$. Therefore, the proposed rule cannot obtain an objective value strictly below $n$. \qed
\end{example}

\subsection{Support of Efficient Rules for Fractional Committee Selection}

\citet{SV24maximumflowfairnetwork} study fractional committees and the lotteries implementing them: a fractional committee $x \in [0,1]^C$ specifies marginal selection probabilities, and a lottery (i.e., a probability distribution over committees) implements $x$ if the probability that a candidate $c \in C$ is in the sampled committee is exactly the marginal $x_c$.
It is easy to see that if a fractional committee $x$ is \fPO{}, then for every implementing lottery, all the committees in its support must themselves be \fPO{} \citep[cf.][Proposition 1]{AFSV23bobw}.
Thus, a natural strategy to prove the existence of fair and efficient committees is to look into the support of a fair and efficient \emph{fractional} committee.

Now, \citet{SV24maximumflowfairnetwork} show that, unlike in the single-winner setting, maximizing Nash welfare over fractional committees does not guarantee their ex-ante fairness notions of group resource proportionality or fractional core. Here we record a complementary ex-post obstruction. Call a feasible fractional committee $x$ \emph{fractional Nash optimal} if it maximizes $\prod_{i\in N} u_i(x)$, or equivalently $\sum_{i\in N}\log u_i(x)$ when all utilities are positive.

\begin{example}[Fractional Nash fails ex-post \JR{}]
\label{prop:Nash_Fails_ExPost_JR}
There exists an instance where every committee in any support of any lottery implementing a fractional Nash optimal committee fails justified representation.

Let $k=4$. Start with $19$ candidates and $\binom{19}{5} = 11628$ voters, each approving a distinct size-$5$ subset of candidates. Add a new candidate for every subset of voters of size $\binom{15}{5} = 3003$, approved exactly by those voters. By symmetry and concavity of the logarithm, the fractional Nash optimal committee assigns mass $4/19$ to each of the original $19$ candidates and mass $0$ to every newly added candidate. Hence, any lottery implementing this fractional committee is supported only on size-$4$ committees of original candidates.

For any such integral committee, exactly $\binom{19-4}{5} = 3003$ voters approve none of its members: these are precisely the voters whose approved $5$-subset of original candidates is disjoint from the selected committee. By construction, there is a newly added candidate approved by exactly this group of voters. Since $3003 \ge n/k = 11628 / 4 = 2907$, every committee in the support violates \JR{}. \qed
\end{example}

There do exist fractional committee rules that are both \fPO{} and satisfy JR-like fairness properties, in particular the \emph{redistributive utilitarian rule} (RUT) of \citet{SV24maximumflowfairnetwork}, as well as the capped Lindahl equilibrium based on the convex program of \citet{KP25lindahl}.
However, \Cref{prop:Nash_Fails_ExPost_JR} also shows that both of those rules fail ex-post \JR{} since in both cases, their implementing lotteries are supported only on size-$4$ committees of original candidates.

\subsection{Repeated Search for Dominating Committee}

A natural procedure for finding a committee that satisfies \JR{} and \PO{} is to start with a committee that satisfies \JR{} and then repeatedly look for a Pareto improvement until we reach a Pareto optimal committee. Because \JR{} is closed under Pareto improvements (and because each Pareto improvement increases utilitarian social welfare by at least 1 utility point), this process will terminate at a committee that is \JR{} and \PO{}. The problem with this strategy is that checking whether a Pareto improvement exists is \NPH{}.

Since fractional Pareto improvements can be found in polynomial time, one might try to run the same idea with fractional improvements. Suppose the current committee is $W$ and a fractional committee $x$ dominates it. If $x$ could be written as a convex combination of \JR{} committees, then the average utilitarian welfare of the committees in this convex combination would be strictly larger than that of $W$. Hence, at least one \JR{} committee in the support would have strictly larger utilitarian welfare than $W$, and we could move to it. Repeating this step would yield a \JR{} committee with no fractional Pareto improvement, i.e., a \JR{}+\fPO{} committee. The following example shows where this approach breaks.

\begin{example}[A \JR{} committee that fails \fPO{} might not be fractionally dominated by any convex combination of \JR{} committees] Take $k = 2$. Let there be $5$ candidates $\{c_1, c_2, c_3, c_4, c_5\}$. The ballots are: $1 \times \{c_1, c_2\}$, $1 \times \{c_1, c_3\}$, $1 \times \{c_1, c_4\}$, $1 \times \{c_2, c_3, c_5\}$, $1 \times \{c_3, c_4, c_5\}$, $6 \times \{c_2, c_4, c_5\}$, and $9 \times \{c_3\}$. The committee $\{c_1, c_5\}$ satisfies \JR{}, but violates \fPO{} (dominated by a committee assigning $0.5$ to each of $c_1, c_2, c_3, c_4$). However, $\{c_1, c_5\}$ is not fractionally dominated by any convex combination of \JR{} committees.
\end{example}

\subsection{Maximizing approval score subject to \JR{}}
The preceding approach only asks whether a fractionally dominated \JR{} committee can be replaced by some \JR{} committee with higher approval score. A stronger approach is to skip the iterative search and directly maximize approval score over all \JR{} committees.
\citet[Corollary 1]{BBF+26fractional} characterize the Thiele rules that always return \fPO{} committees: for an objective of the form $\sum_{i\in N} h(u_i(W))$, this holds exactly when $h$ is strictly increasing and convex, where convexity means that $h(u+1)-h(u)\ge h(u)-h(u-1)$ for every $u\in\mathbb{N}$. Convex Thiele objectives reward concentrating utility, so they are not natural candidates for \JR{}-style fairness; the linear boundary case, Approval Voting, is the most fairness-friendly member of this family. Since Approval Voting itself can fail \JR{}, one might instead maximize approval score subject to \JR{}. The selected committee is \JR{} by construction and is also \PO{}, because any Pareto improvement would preserve \JR{} and strictly increase approval score. The following counterexample shows that this constrained rule can still fail \fPO{}.

\begin{example}[Maximizing approval score subject to \JR{} can violate \fPO{}]
    Take $k = 3$. Consider $n = 100$ voters and $6$ candidates $\{c_1, c_2, \ldots, c_6\}$. The voters are as follows:
    \begin{itemize}
        \item $33$ voters approve $\{c_5\}$.
        \item $1$ voter each approves $\{c_1, c_4, c_5\}$, $\{c_2, c_4, c_5\}$, and $\{c_3, c_4, c_5\}$.
    \item $51$ voters approve $\{c_1, c_2, c_3, c_4, c_6\}$.
    \item $4$ voters each approve $\{c_1, c_2, c_6\}$, $\{c_1, c_3, c_6\}$, and $\{c_2, c_3, c_6\}$.
    \item $1$ voter approves $\{c_6\}$.
    \end{itemize}
    Here, $c_1, c_2$ and $c_3$ are approved by $60$ voters each, $c_4$ is approved by $54$ voters, $c_5$ is approved by $36$ voters and $c_6$ is approved by $64$ voters. The unique committee that maximizes the approval score subject to being \JR{} is $W^\star = \{c_1, c_2, c_3\}$. $W^\star$ is \JR{} since the only unrepresented voters are the $33 < 100/3$ voters who approve only $\{c_5\}$ and the one voter who approves only $\{c_6\}$. Additionally, the only committees with a (weakly) higher approval score are $\{c_1, c_2, c_6\}, \{c_1, c_3, c_6\}$ and $\{c_2, c_3, c_6\}$. However, in each of these committees $c_5$ is a candidate that is approved by $33+1 = 34 \geq 100/3$ unrepresented voters. 

    Now, consider the fractional committee $y$, with $y_{c_1} = y_{c_2} = y_{c_3} = 0.9, y_{c_4} = 0.1, y_{c_5} = 0, y_{c_6} = 0.2$. Note that $u_i(y) \geq u_i(W^\star)$ for all voters $i$ and at least one inequality is strict:
    \begin{itemize}
        \item The voters who approve $\{c_5\}$ derive a utility of $0$ each from either $W^\star$ or $y$.
        \item The voters who approve exactly $\{c_1, c_4, c_5\}$, $\{c_2, c_4, c_5\}$, or $\{c_3, c_4, c_5\}$ derive a utility of $1$ each from either $W^\star$ or $y$.
        \item The voters who approve $\{c_1, c_2, c_3, c_4, c_6\}$ derive a utility of $3$ each from either $W^\star$ or $y$.
        \item The voters who approve $\{c_1, c_2, c_6\}$, $\{c_1, c_3, c_6\}$, or $\{c_2, c_3, c_6\}$ derive a utility of $2$ each from either $W^\star$ or $y$.
        \item The voter who approves only $\{c_6\}$ derives a utility of $0$ from $W^\star$ and $0.2$ from $y$. \qed
    \end{itemize}
\end{example}

\subsection{Maximizing a Thiele Objective Subject To \fPO{}}

All Thiele methods that satisfy \JR{} fail \fPO{} \citep[Observation 2]{BBF+26fractional}. Still, one might hope that maximizing, for example, the PAV or CC score across all \fPO{} committees might still allow the proof for \JR{} to go through. But this is not the case.

\begin{example}[Maximizing PAV score subject to \fPO{} can violate \JR{}]
	Take $k=8$. Let $D=\{c_1,c_2,c_3,c_4\}$ and $P=\{p_1,\dots,p_7\}$. The candidate set is $\{a,b\}\cup D\cup P$, so $m=13$. The voters are as follows:
	\begin{itemize}
		\item $9$ voters with approval set $\{a\}\cup S$ for each $S\in\binom{D}{2}$,
		\item $8$ voters with approval set $\{b\}\cup S$ for each $S\in\binom{D}{2}$,
		\item 1 voter with approval set $\{a\}\cup D$,
		\item $281$ voters with approval set $P$.
	\end{itemize}
	Thus, $n=6\cdot 9+6\cdot 8+1+281=384$. Consider the committee $W = \{a\} \cup P$. Note that $W$ is \fPO{} since it is the unique committee maximizing the approval score. This is because each $p \in P$ is approved by $281$ voters, $a$ is approved by $9 \cdot 6 + 1 = 55$ voters, $b$ is approved by $8 \cdot 6 = 48$ voters, and each $d \in D$ is approved by $9 \cdot 3 + 8 \cdot 3 + 1 = 52$ voters. Additionally, $W$ is the unique PAV-maximizing committee among \fPO{} committees. First, note that for each $p \in P$, the committee $U_p \coloneqq \{a, b\} \cup (P \setminus \{p\})$ is not \fPO{}. Indeed, it is fractionally dominated by the fractional committee assigning weight $1$ to each candidate in $P\setminus\{p\}$ and weight $\frac12$ to each candidate in $D$. The voters approving $P$ keep utility $6$, all voters approving $\{a\}\cup S$ or $\{b\}\cup S$ keep utility $1$, and the voter approving $\{a\}\cup D$ improves from utility $1$ to utility $2$.
    
    It remains to argue that every other committee has a strictly lower $\PAV$ score. We have $\PAV(W) = 281 H_7 + 55$, $\PAV(\{b\} \cup P) = 281 H_7 + 48$ and $\PAV(\{d\} \cup P) = 281 H_7 + 52$. All remaining committees of size $8$ must have at least $2$ candidates from $\{a, b\} \cup D$. Formally, consider any remaining committee $T = R \cup Q$, where $R \subseteq \{a, b\} \cup D$ and $Q \subseteq P$. Then, $R \neq \{a, b\}$ and $r = |R| \geq 2$ (the cases $R = \{a, b\}, \{a\}, \{b\}$ and $\{d\}$ for some $d \in D$ have already been considered). Let $M(R)$ denote the contribution of the $103$ voters not approving exactly $P$. The following table gives, for each $r$, the maximum possible value of $M(R)$ and the resulting margin by which $W$ beats the best such committee:
	\[
	\begin{array}{c|c|c}
		r & \max M(R) & \operatorname{PAV}(W)-\max_T \operatorname{PAV}(T)\\
		\hline
		2 & 95 & 1/7\\
		3 & 388/3 & 177/14\\
		4 & 1861/12 & 6033/140\\
		5 & 10397/60 & 666/7\\
		6 & 11357/60 & 3629/21
	\end{array}
	\]
	All margins are strictly positive. Thus, the only committees with PAV score above $W$ are the committees $U_p$, and all of them fail \fPO{}. However, $W$ violates \JR{}. The $48$ voters approving some set $\{b\}\cup S$, $S\in\binom{D}{2}$, all approve $b$ and approve no member of $W$. Since $48=\frac{384}{8}=\frac{n}{k}$,  these voters witness a \JR{} violation.
\qed
\end{example}

\begin{example}[Maximizing CC score subject to \fPO{} can violate \JR{}]
	Take $k=4$. Let $D=\{c_1,c_2,c_3,c_4,c_5,c_6\}$ and $P=\{p_1,p_2,p_3\}$. The candidate set is $\{a,b\}\cup D\cup P$, so $m=11$. The voters are as follows:
	\begin{itemize}
		\item $4$ voters with approval set $\{a\}\cup S$ for each $S\in\binom{D}{3}$,
		\item $3$ voters with approval set $\{b\}\cup S$ for each $S\in\binom{D}{3}$,
		\item $1$ voter with approval set $\{a\}\cup D$,
		\item $33$ voters with approval set $\{p_j\}$ for each $p_j \in P$.
	\end{itemize}
	Thus, $n=4\cdot 20+3\cdot 20+1+ 33 \cdot 3=240$. Consider the committee $W = \{a\} \cup P$. Give weight 3 to the singleton voters, and weight 1 to other voters. Then each $p_j$ has score $99$ and $a$ has score $81 = 4 \cdot 20 + 1$, with other candidates having lower score. Thus $W$ is \fPO{}. Additionally, $W$ is the unique CC-maximizing committee among \fPO{} committees. The CC score of $W$ itself is $180 = 4\cdot 20 + 1 + 33 \cdot 3$. Let us compute the CC score of other committees of size $k = 4$. First note that $b$ is a worse choice than $a$, so committees that contain $b$ but not $a$ can be improved by swapping the two.
	\begin{itemize}
		\item $W' = \{a, b, d_i, d_j\}$ has score $141$.
		\item $W' = \{a, b, d_i, p_j\}$ has score $174$.
		\item $W' = \{a, b, p_i, p_j\}$ has score $207$.
		\item $W' = \{a, d_i, p_i, p_j\}$ has score $81 + 3 \cdot \binom{|D|-1}{2} + 66 = 177$.
		\item $W' = \{a, d_i, d_j, p_k\}$ has score $81 + 3 \cdot (\binom{6}{3} - \binom{4}{3}) + 33 = 162$.
		\item $W' = \{a, d_i, d_j, d_k\}$ has score $81 + 3 \cdot (\binom{6}{3} - \binom{3}{3})  = 138$.
		\item Mixtures of $D$ and $P$ have score at most $179$.
	\end{itemize}
	Hence, the only committees that beat $W$ on CC score are of the form $\{a, b, p_i, p_j\}$. But these committees fail \fPO{} because they are dominated by the fractional committee with weight 1 on $p_i$ and $p_j$, and with weight $\frac13$ on each candidate from $D$ (every voter keeps their previous utility, except the voter with approval set $\{a\} \cup D$ who improves from 1 to 2).
	
	However, $W$ fails \JR{} because $\frac{n}{k} = 60$ voters approve $b$ while they are unrepresented in $W$.
	\qed
\end{example}

\end{document}